\documentclass[a4paper,10pt]{article}
\usepackage{array}
\usepackage{longtable}
\usepackage{multirow}
\usepackage{hhline}

\usepackage[dvips]{graphicx}

\usepackage{amsfonts,amssymb, amsmath, amsthm}
\usepackage[english]{babel}

\newcommand {\ts} {\textsl}

\def\sgn{\mathop{\rm sgn}\nolimits}
\def\bbR{\mathbb{R}}
\def\bbZ{\mathbb{Z}}

\newtheorem{theorem}{Theorem}
\newtheorem{corol}{Corollary}
\newtheorem{propos}{Proposition}
\newtheorem{lemma}{Lemma}

\newtheorem{agre}{Agreement}

\newcommand{\colu}[2]{\begin{pmatrix} #1 \\ #2 \end{pmatrix}}
\newcommand{\matr}[4]{\begin{pmatrix} #1 & #2 \\ #3 & #4 \end{pmatrix}}
\def\inn{^{\rm in}}
\def\out{^{\rm out}}
\def\iso{I}

\def\tG{\Delta}

\def\QR{Q}
\def\PR{P}

\def\Tor{{\rm \bf T}}
\def\Cir{{\rm \bf S}}

\newcommand{\fns}[1]{{\footnotesize #1}}
\newcommand {\fts}[1] {{\small #1}}
\newcommand{\rul}{\rule[-5pt]{0pt}{18pt}}

\numberwithin{equation}{section}

\usepackage{hyperref}

\begin{document}

\author{M.P. Kharlamov\\ \\
\it Russian Academy of National Economy and Public Administration\\
\it Volgograd Branch, Volgograd, Russia}

%\ead{mikeh@inbox.ru}
%\ead[url]{www.vags.ru/pp/hmp}
\title{Phase topology of one system with separated variables\\and singularities of the symplectic structure\footnote{J. of Geometry and Physics, On-line July 2014, DOI: 10.1016/j.geomphys.2014.07.004}}

\date{}

\maketitle

\begin{abstract}
We consider an example of a system with two degrees of freedom admitting separation of variables but having a subset of codimension 1 on which the 2-form defining the symplectic structure degenerates. We show how to use separation of variables to calculate the exact topological invariant of non-degenerate singularities and singularities appearing due to the symplectic structure degeneration. New types of non-orientable 3-atoms are found.



\end{abstract}

%\linenumbers
%\tableofcontents
%% main text
\section{Introduction}\label{sec1}
The Lie co-algebra $e(3,2)^*$ with coordinate functions $g_i,\alpha_j,\beta_k$ has the Lie\,--\,Poisson bracket
\begin{equation}\notag
\begin{array}{c}
\{g_i,g_j\}=-\varepsilon_{ijk}g_k,\quad \{g_i,\alpha_j\}=-\varepsilon_{ijk}\alpha_k,\quad \{g_i,\beta_j\}=-\varepsilon_{ijk}\beta_k,\\
\{\alpha_i,\alpha_j\} = \{\beta_i,\beta_j\} = \{\alpha_i,\beta_j\} = 0 \qquad (i,j,k=1,2,3).
\end{array}
\end{equation}
For a given function $H:e(3,2)^* \to \bbR$ the system of equations written as
\begin{equation}\label{eq1_1}
  \dot x = \{x,H\}
\end{equation}
is called Euler equations on $e(3,2)^*$ with the Hamilton function $H$ \cite{Bogo}.

Suppose we have a rigid body rotating about a fixed point $O$ and relate all
vector and tensor objects to a reference frame moving with the body.
Let $\boldsymbol{\alpha},\boldsymbol{\beta}$ be some vectors fixed in the inertial space and $\mathbf{g}$ the kinetic momentum of the body. The equations of motion have the form \eqref{eq1_1} with $H=\frac{1}{2} \mathbf{g}\cdot \mathbf{I}^{-1}\mathbf{g} + W(\boldsymbol{\alpha},\boldsymbol{\beta})$.
The constant symmetric matrix $\mathbf{I}$ is the inertia tensor and $\boldsymbol{\omega}=\mathbf{I}^{-1}\mathbf{g}$ is the angular velocity of the body. The function $W$ is treated as the potential energy. In what follows we use the coordinates $\omega_i$ of $\boldsymbol{\omega}$ in the moving frame instead of $g_i$ for convenience. In the generic case $\boldsymbol{\alpha}{\times}\boldsymbol{\beta} \ne 0$, common levels $\mathcal{P}$ of the Casimir functions $\boldsymbol{\alpha}^2, \boldsymbol{\beta}^2, \boldsymbol{\alpha}{\cdot}\boldsymbol{\beta}$ are 6-dimensional symplectic leaves of the Lie\,--\,Poisson bracket. System \eqref{eq1_1} restricted to $\mathcal{P}$ becomes a Hamiltonian system with three degrees of freedom.

The Hamilton function
\begin{equation}\label{eq1_2}
H = \omega_1^2 + \omega_2^2 + \dfrac{1}{2}\omega_3^2 - \alpha_1 -\beta_2
\end{equation}
defines an integrable generalization of the Kowalevski top \cite{Kowa} to a double force field. Additional integrals in involution are \cite{Bogo,ReySem}:
\begin{equation}\notag
\begin{array}{l}
K = ( \omega_1^2 - \omega_2^2 + \alpha_1 - \beta_2 )^2 + (2 \omega_1 \omega_2 + \alpha_2 + \beta_1)^2,\\
G = (\omega_1 \alpha_1+\omega_2 \alpha_2 + \frac{1}{2}\omega_3 \alpha_3)^2  +(\omega_1 \beta_1+\omega_2 \beta_2 + \frac{1}{2}\omega_3 \beta_3)^2 \\
\qquad +\omega_3 [ \omega_1 (\alpha_2 \beta_3-\alpha_3 \beta_2)+\omega_2 (\alpha_3 \beta_1-\alpha_1 \beta_3) + \frac{1}{2}\omega_3 (\alpha_1 \beta_2-\alpha_2 \beta_1)]\\
\qquad -\alpha_1 \boldsymbol{\beta}^2-\beta_2 \boldsymbol{\alpha}^2+(\alpha_2+\beta_1)(\boldsymbol{\alpha}{\cdot}\boldsymbol{\beta}).
\end{array}
\end{equation}

It is known that for a wide class of Hamilton functions including \eqref{eq1_2} without loss of generality one can choose the following constants of the Casimir fun\-ctions~\cite{KhRCD}
\begin{equation}\label{eq1_3}
  \boldsymbol{\alpha}^2=a^2, \qquad \boldsymbol{\beta}^2=b^2, \qquad \boldsymbol{\alpha}{\cdot}\boldsymbol{\beta}=0 \qquad (a > b >0).
\end{equation}
For the function \eqref{eq1_2} the cases $b=0$ and $b=a$ correspond to the classical Kowalevski case \cite{Kowa} and the case of Yehia \cite{Yeh1}. Both of these cases have symmetries, globally reduce to systems with two degrees of freedom and are not considered here. For irreducible cases \eqref{eq1_3} the system \eqref{eq1_1}, \eqref{eq1_2} admits a Lax representation given by A.G.~Reyman and M.A.~Semenov-Tian-Shansky \cite{ReySem} but has not been yet reduced to quadratures. Let us call this system \textit{the RS-system}.

The study of irreducible integrable 3D-systems begins with detecting the so-called critical subsystems. These are even-dimensional invariant submanifolds of the phase space with the induced Hamiltonian systems having less than three degrees of freedom. All critical subsystems for the RS-system were found in \cite{Bogo,Odin,KhRCD}. Separation of variables is known for two of them \cite{KhSavUMBeng,KhRCD07}. Consider a critical subsystem on a four-dimensional invariant submanifold. The 2-form defining the Hamiltonian type of the induced dynamics is obtained as the restriction of the symplectic structure of $\mathcal{P}$. It appeared that, in the RS-system, all critical subsystems with two degrees of freedom have 3-dimensional subsets on which this 2-form degenerates. Such systems are now called \textit{almost} Hamiltonian. For the first subsystem \cite{Bogo}, as shown in \cite{Zo2006}, the Fomenko\,--\,Zieschang invariant \cite{FoZi1991} can be applied. Here we study the second subsystem found in \cite{Odin}. It is denoted by $\mathcal{N}$. This notation stands for the dynamical system and therefore includes both the phase space and the induced dynamics. For the sake of brevity we also call $\mathcal{N}$ the phase space of the subsystem meaning the corresponding subset of $\mathcal{P}$. The rough phase topology of the system $\mathcal{N}$ was described in \cite{KhSavUMBeng}. Nevertheless, some properties of $\mathcal{N}$ has not been completely established and the character of some ``strange'' bifurcations has not been explained. In this paper we fulfil the complete topological investigation of the system  $\mathcal{N}$ in terms of the topological invariants \cite{FoZi1991,BolMet,igs} calculated using the global separation of variables. New topological effects are revealed due to non-orientable bifurcations, which are possible in almost Hamiltonian systems.

%\clearpage

\section{The main system and separation of variables}\label{sec2}
We use the definition of $\mathcal{N}$ given in \cite{KhSavUMBeng}. Consider the first integral of the RS-system
\begin{equation}\notag
F = (2G - p^2 H)^2 - r^4 K,
\end{equation}
where $p = \sqrt{a^2 + b^2}$, $r = \sqrt{a^2 - b^2}$ $(p>r>0)$. As shown in \cite{KhRCD}, the equation $F=0$ applied to the integral constants gives a leaf of the bifurcation diagram of the global integral map $H{\times}K{\times}G: \mathcal{P} \to \bbR^3$. Therefore, we define a non-empty invariant set as follows
\begin{equation}\notag
  \mathcal{N} = \{x \in \mathcal{P}: F(x)=0, \, {\rm d} F(x)=0 \}.
\end{equation}
The invariant submanifold $\mathcal{N}\subset \mathcal{P}$ was first found in \cite{Odin} in terms of invariant relations on $\mathcal{P}$ having singularities on the set $\{\Lambda=0\}$, where
$$
\Lambda = (\alpha_1-\beta_2)^2+(\alpha_2+\beta_1)^2.
$$
The Lie\,--\,Poisson bracket $L$ of these relations is defined as
\begin{equation*}
  L =\frac{1}{\sqrt{\Lambda}}\left\{\omega_1^2+\omega_2^2+[a^2+b^2-2(\alpha_1 \beta_2-\alpha_2 \beta_1)] M\right\}, \quad M=\dfrac{1}{r^4}(2G-p^2 H).
\end{equation*}
It is proved in \cite{KhSavUMBeng} that in the domain $\{\Lambda \ne 0\}$ the set $\mathcal{N}$ is a smooth \mbox{4-d}i\-men\-sio\-nal manifold, $L$ is a partial integral on $\mathcal{N}$ and the restriction to $\mathcal{N}$ of the symplectic structure is non-degenerate everywhere except for the subset ${\{L=0\}}$. This subset is non-empty. Moreover, it does not contain critical points of $L$ and therefore is a 3-dimensional submanifold in $\mathcal{N}$. Thus, the induced system on $\mathcal{N}$ is almost Hamiltonian, i.e., the 2-form defining the Hamiltonian field degenerates on a subset of codimension 1. Below we give some simple explanation of the fact that $\mathcal{N}$ is everywhere a smooth 4-dimensional manifold. Moreover, it appears to be non-orientable. In the sequel, we call $\mathcal{N}$ a manifold without further comments on the fact.

From now on having a first integral denoted by an upper case letter we denote its arbitrary constant by the corresponding lower case letter. On $\mathcal{N}$, the following identities hold \cite{KhSavUMBeng}
\begin{eqnarray}
& g = \dfrac{1}{2}(p^2 h + r^4 m), \qquad k = r^4 m^2, \label{eq2_1}\\
& \ell^2 = 1 + 2 m h + 2 p^2 h^2. \label{eq2_2}
\end{eqnarray}
These relations show that it is convenient to take $M,H$ for the functionally independent pair of the first integrals on $\mathcal{N}$, while the pairs $L,H$ and $L,M$ are not in one-to-one correspondence with the triple $(H,K,G)$. Thus, we define the integral map
$$
J=M{\times}H: \mathcal{N} \to \bbR^2
$$
and study the bifurcations of the integral manifolds
\begin{equation*}
  J_{m,h} = \{x \in \mathcal{N}: M(x)=m,\, H(x)=h\}.
\end{equation*}
For almost all integral manifolds $J_{m,h}$ we can take two values of $\ell$ with different signs according to \eqref{eq2_2}.

Introduce the new variables $u_1,u_2$ as
\begin{equation}\notag
u_1 = \dfrac{a \sqrt{\Lambda}}{a^2-(\alpha_1\beta_2-\alpha_2\beta_1)},\quad u_2 = \dfrac{b^2-(\alpha_1\beta_2-\alpha_2\beta_1)}{b \sqrt{\Lambda}}.
\end{equation}
It readily follows from \eqref{eq1_3} that
\begin{equation}\label{eq2_3}
|u_1|\leqslant 1, \qquad |u_2| \leqslant 1.
\end{equation}
These are the so-called natural restrictions.
Let
\begin{equation*}
\begin{array}{c}
  \tau_{1,2} = \dfrac{2am}{\ell \mp 1}, \quad \sigma_{1,2} = \dfrac{\ell \mp 1}{2 b m }, \quad \Theta = \dfrac{1}{a -b \, u_1 u_2},\\
  h_*=-\dfrac{1}{2}(h+p^2m), \quad   \psi= a b \,m \, u_2+ h_* u_1-\frac{1}{2}(a + b\,u_1 u_2).
\end{array}
\end{equation*}
Consider the two-valued (algebraic) radicals
\begin{equation}\label{eq2_4}
\begin{array}{rl}
  \QR_1=\sqrt{1-u_1^2}, & \QR_2=\sqrt{1-u_2^2}, \\
  \PR_1=\sqrt{\mathstrut h_*(u_1-\tau_1)(u_1-\tau_2)},  & \PR_2=b \sqrt{\mathstrut m (u_2-\sigma_1)(u_2-\sigma_2)}.
\end{array}
\end{equation}
\begin{propos}\label{prop02}
The differential equations induced on $\mathcal{N}$ by the RS-system separate in variables $u_1,u_2$,
\begin{equation}\label{eq2_5}
  \dot u_1 = \QR_1 \PR_1, \qquad \dot u_2 = \QR_2 \PR_2,
\end{equation}
and on the integral manifolds $J_{m,h}$ the phase variables have the following expressions in terms of $u_1,u_2$
\begin{equation}\label{eq2_6}
\renewcommand{\arraystretch}{1.1}
\begin{array}{c}
\displaystyle \alpha _1  = -{2a } {\Theta^{2}}[(a u_1-b u_2)\psi + b \QR_1 \QR_2 \PR_1 \PR_2], \\
\displaystyle \alpha _2  = \phantom{-}{2a } {\Theta^{2}}[ (a u_1-b u_2) \PR_1 \PR_2 -  b \QR_1 \QR_2 \psi], \\
\displaystyle \beta _1  =  \phantom{-}{2b } \,{\Theta^{2}}[a \QR_1 \QR_2 \psi - (b u_1 - a u_2) \PR_1 \PR_2], \\
\displaystyle \beta _2  =  \phantom{-}{2b } \,{\Theta^{2}}[a \QR_1 \QR_2 \PR_1 \PR_2- (b u_1-au_2)\psi],\\
\displaystyle \omega _1  = {r (\ell u_1-2 a m)} {\Theta^{}}\PR_2,\quad \omega _2  = {r (2b m u_2-\ell)} {\Theta^{}}\PR_1,\\
\displaystyle \alpha _3  = {a r} {\Theta^{}}\QR_1 , \quad \beta _3  ={ b r} {\Theta^{}}u_1 \QR_2,  \quad \omega _3  = {2}{\Theta^{}}(b \QR_2 \PR_1 - a \QR_1 \PR_2).
\end{array}
\end{equation}
\end{propos}

The result easily follows from \cite{KhSavUMBeng}. The variables of separation $s_1,s_2$ found in \cite{KhSavUMBeng} can oscillate on infinite segments. Therefore here we use dimensionless variables $u_1=a s_1^{-1}, u_2 =b^{-1}s_2$, always restricted to the segment $[-1,1]$. Let us note that, according to \eqref{eq1_3} and \eqref{eq2_3}, $a-b u_1 u_2>0$ everywhere, so the factor $\Theta$ in \eqref{eq2_6} is well defined and always positive.

\section{Geometrical representation of the integral manifolds}\label{sec3}
Fixing the values $m,h$ and choosing $\ell$ according to \eqref{eq2_2}, let us treat \eqref{eq2_6} as a map
\begin{equation*}
  \pi: V^6=\bbR^3(u_1,\QR_1,\PR_1){\times}\bbR^3(u_2,\QR_2,\PR_2)\to \mathcal{N}.
\end{equation*}
Then the integral manifold $J_{m,h}$ is the $\pi$-image of the direct product of two curves $\Gamma_i \subset \bbR^3(u_i,\QR_i,\PR_i)$ ($i=1,2$)
\begin{equation}\label{eq3_1}
\begin{array}{l}
  \Gamma_1: \left\{ \begin{array}{l}
                          \QR_1^2+u_1^2=1,\\
                          \PR_1^2-f_1(u_1)=0;
                     \end{array}
                          \right. \qquad   \Gamma_2: \left\{ \begin{array}{l}
                            \QR_2^2+u_1^2=1, \\
                            \PR_2^2-f_2(u_2)=0,
                     \end{array}
                            \right.
\end{array}
\end{equation}
where
\begin{equation*}
\begin{array}{l}
  f_1(u_1) =h_* u_1^2 + a \ell u_1 - a^2 m = h_*(u_1-\tau_1)(u_1-\tau_2) \\
  f_2(u_2) = b^2 m u_2^2 - b \ell u_2 -h_* = b^2 m (u_2-\sigma_1)(u_2-\sigma_2) .
\end{array}
\end{equation*}
Introduce the enhanced space $V^9=\bbR^3(\ell,m,h){\times}\bbR^3(u_1,\QR_1,\PR_1){\times}\bbR^3(u_2,\QR_2,\PR_2)$,
and define $\mathcal{M} \subset V^9$ by equations \eqref{eq2_2} and \eqref{eq3_1}. Let
$\hat{\pi}: \mathcal{M} \to \mathcal{N}$ be the map given by \eqref{eq2_6}.

\begin{lemma}\label{lem1}
The set $\mathcal{M}$ is a connected smooth 4-dimensional manifold in $V^9$.
\end{lemma}
\begin{proof}
It is easy to check that the system of five equations \eqref{eq2_2}, \eqref{eq3_1} always has rank~5. Therefore, $\mathcal{M}$ is a smooth 4-dimensional manifold.

Let us call a point $(\ell,m,h)\in \bbR^3$ admissible if the corresponding set
\begin{equation}\notag
\tG=\tG(\ell,m,h)=\Gamma_1{\times}\Gamma_2
\end{equation}
is not empty. For a given $\tG$, we call the projection of it onto the $(u_1,u_2)$-plane a region of possible motions. If not empty, this region is a rectangle in the square \eqref{eq2_3} cut out by the system of inequalities $\{f_1(u_1) \geqslant 0, \; f_2(u_2) \geqslant 0\}$. This yields that the image of the integral map $J$ in the $(m,h)$-plane is
\begin{equation}\notag
\begin{array}{l}
{\rm Im}\,J =\{h \geqslant \min [r^2 m - 2a, - r^2 m - 2b], \; 2p^2 m^2+2 h m +1 \geqslant 0\}.
\end{array}
\end{equation}
The set $\mathcal{D}$ of all admissible points is a two-sheet covering of ${\rm Im}\,J$. The sheets with opposite signs of $\ell$ are glued together along the curve
\begin{equation}\label{eq3_2}
  L_0: \; 2p^2 m^2+2 h m +1=0, \quad \ell=0, \quad m <0
\end{equation}
and $\mathcal{D}$ is connected. Take any two points in $\mathcal{M}$ and connect their images in $\mathcal{D}$ by a path $\gamma$ with ${\rm Int}\,\gamma \subset {\rm Int}\mathcal{D}$. At the points of ${\rm Int}\,\gamma$ the curves $\Gamma_i$ never degenerate to one point (we omit technical details, since this fact follows from the rough topological analysis given below). Therefore the initial two points in $\mathcal{M}$ can also be connected by a path.
\end{proof}

\begin{lemma}\label{lem3}
Consider the involution $\chi : V^9 \to V^9$ $(\chi^2={\rm Id})$
\begin{equation*}
  \chi : (\ell,m,h,u_1,\QR_1,\PR_1,u_2,\QR_2,\PR_2) \mapsto (-\ell,m,h,-u_1,\QR_1,-\PR_1,-u_2,-\QR_2,\PR_2).
\end{equation*}
The manifold $\mathcal{M}$ and the map $\hat\pi$ are $\chi$-invariant and $\chi$ changes the orientation on $\mathcal{M}$.
\end{lemma}
Indeed, for $x\in \mathcal{M}$ let $\mu_x \in L(5,9)$ denote the Jacobi matrix of the system \eqref{eq2_2}, \eqref{eq3_1}. We readily obtain that $\mu_{\chi(x)} \equiv \mu_x \cdot A_\chi$, where $A_\chi$ is the matrix of $\chi$. So $\chi$ changes the orientation of $V^9$ but preserves the orientation of the space normal to $\mathcal{M}$. Therefore, it changes the orientation of $\mathcal{M}$.

Let us emphasize that the regions of possible motions for $(\ell,m,h)$ and $(-\ell,m,h)$ are centrally symmetric to each other and do not coincide except for $\ell=0$. In the latter case $\chi$ becomes a $\bbZ_2$-symmetry of $\tG$. In particular, $\chi:\mathcal{M} \to \mathcal{M}$ is a diffeomorphism and $\mathcal{N}=\mathcal{M}/\chi$. Summarizing the above statements we come to the following theorem.

\begin{theorem}\label{theo1}
The set $\mathcal{N}$ is a smooth connected 4-dimensional non-orientable submanifold in $\mathcal{P}$.
For any point $(m,h)\in ({\rm Im}\,J)\backslash L_0$ the integral manifold $J_{m,h}$ is diffeomorphic to $\Gamma_1{\times}\Gamma_2$. On $L_0$ we have $J_{m,h} \cong (\Gamma_1{\times}\Gamma_2)/\bbZ_2$.
\end{theorem}

\begin{figure}[!ht]
%figure01%
\centering
\includegraphics[width=0.3\textwidth, keepaspectratio = true]{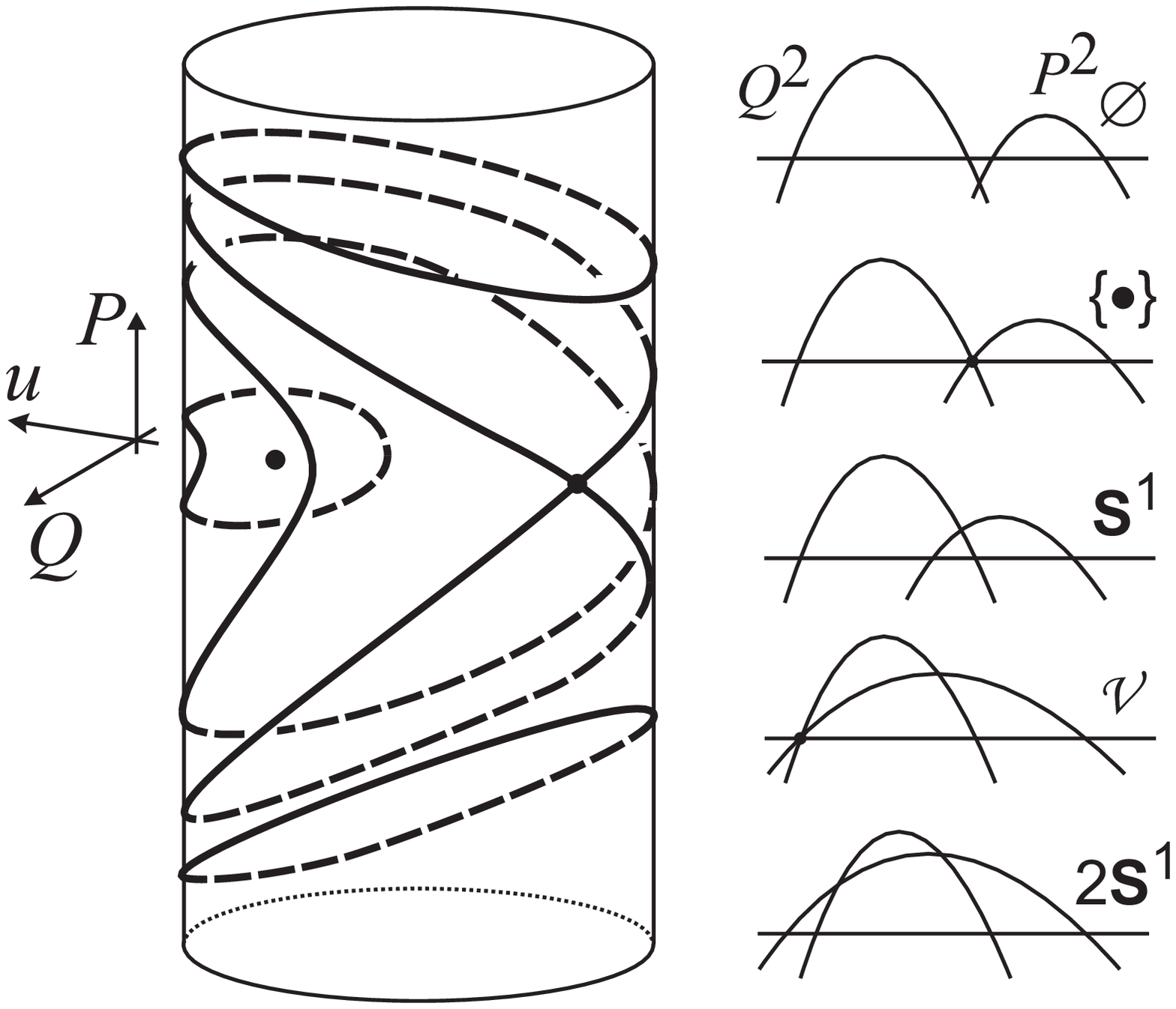}\
\includegraphics[width=0.3\textwidth, keepaspectratio = true]{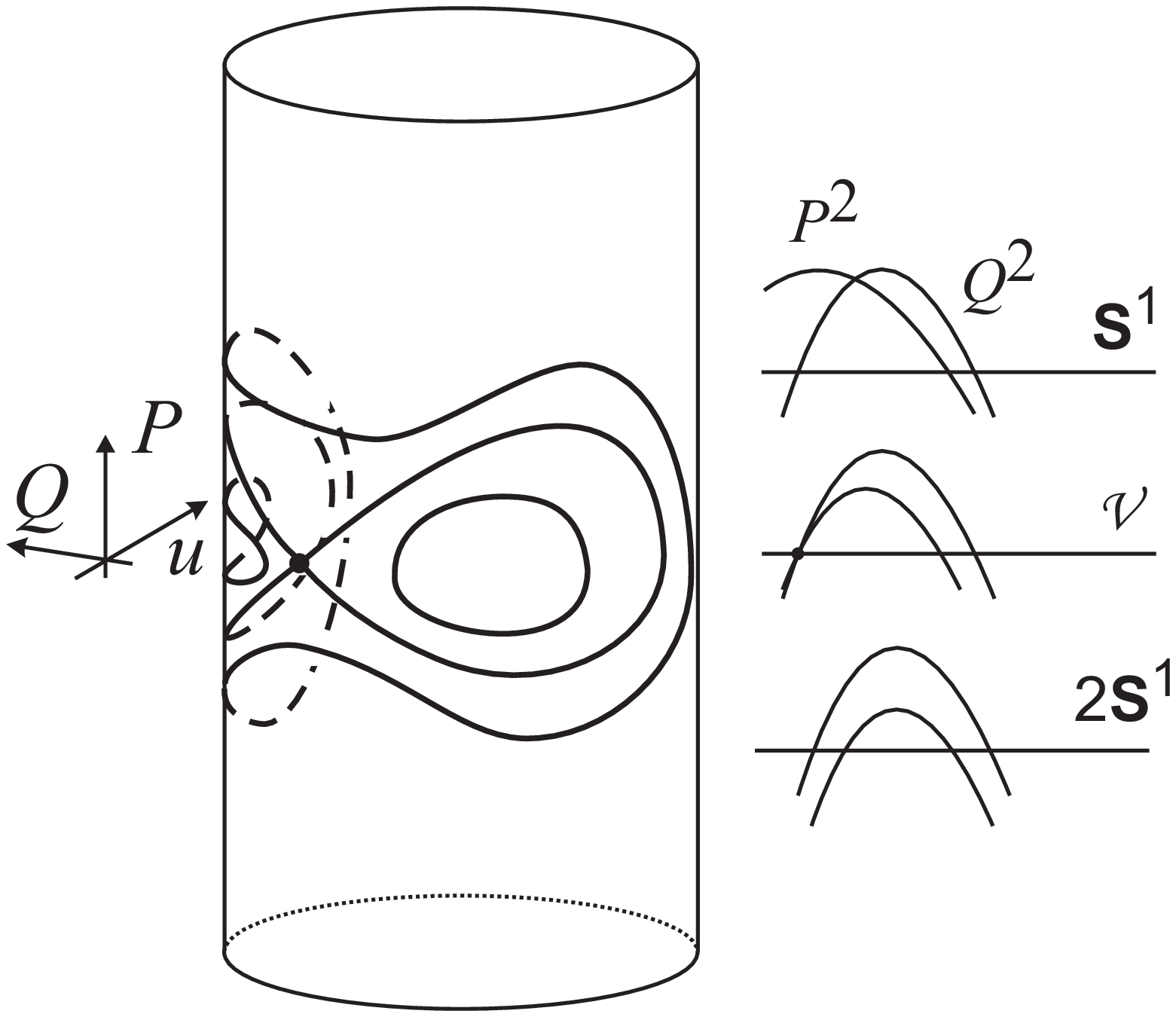}\
\includegraphics[width=0.3\textwidth, keepaspectratio = true]{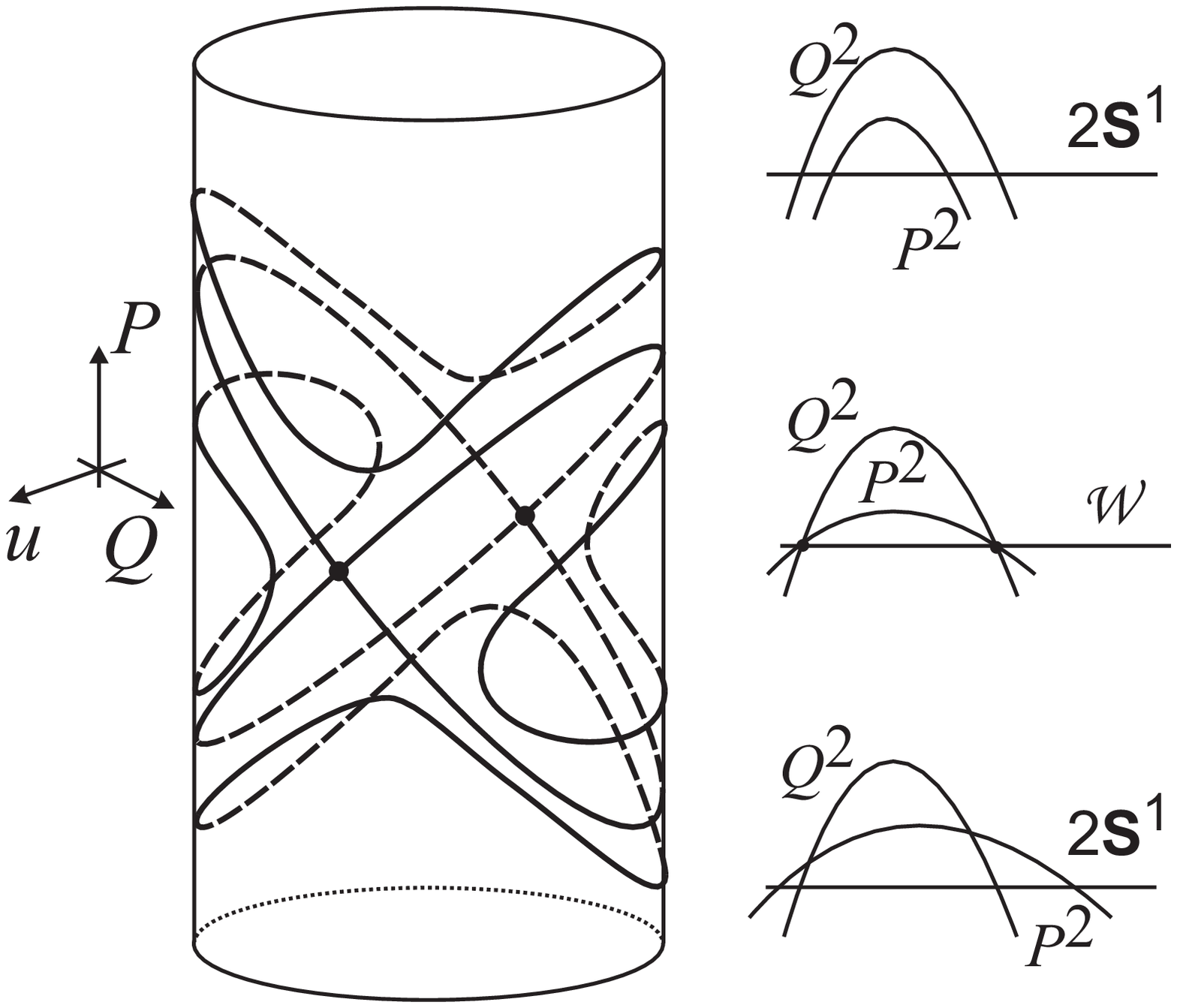}
\caption{Possible transformations of $\Gamma$.}\label{figgamma}
\end{figure}

As we can see from \eqref{eq3_1}, the set $\Gamma_i$ in $\bbR^3(u_i,\QR_i,\PR_i)$ is an intersection of the round cylinder generated by the unit circle in the $(u_i,\QR_i)$-plane and a cylinder with elliptical directrix in the $(u_i,\PR_i)$-plane. The form of $\Gamma_i$ depends on the position of the roots of $f_i(u_i)$ with respect to $\pm 1$. The standard transformations which can happen to $\Gamma_i$ are shown in Fig.~\ref{figgamma}:

$1)$ $\varnothing \to \{\boldsymbol{\cdot}\} \to \Cir^1 \to \mathcal{V} \to 2\Cir^1$;

$2)$ $\Cir^1 \to \mathcal{V} \to 2\Cir^1$;

$3)$ $2\Cir^1 \to \mathcal{W} \to 2\Cir^1$.

\noindent Here $\mathcal{V}$ and $\mathcal{W}$ stand respectively for the eight-curve $\Cir^1 \dot{\cup} \Cir^1$ with transversal self-intersection and for the curve $\Cir^1 \ddot{\cup} \Cir^1$ formed by two circles transversally intersecting at two points. The difference between cases $1$ and $2$ is that the two circles of $1$ differ by the sign of $\QR_i$, and the two circles of $2$ differ by the sign of $\PR_i$. In the first case there is an exit to $\varnothing$ when the common root of $\PR^2_i$ and $\QR^2_i$ turns out to be an isolated point. In the second case both exits would be of hyperbolic type because $\QR^2_i$ never has a multiple root in this system.

\section{Rough topology}\label{sec4}
To accomplish the \textit{rough} topological analysis of an integrable system one has to do the following:

1) find the admissible region (the image of the integral map);

2) for each regular value of the integral map find the number of Liouville tori in the pre-image of this value, i.e., establish the topology of regular integral manifolds;

4) for each critical value of the integral map find the topological type of its pre-image, i.e.,  establish the topology of irregular integral manifolds usually called (critical) integral surfaces;

5) show a collection of paths in the integral constants space with complete description of the topological type of bifurcations taking place along this path; this set of paths should be sufficient to find out the character of any bifurcation occurring in the phase space.

For the classical integrable systems in the rigid body dynamics, this program was fulfilled in \cite{Kh1976,Kh1979,PoKh1979,KhPMM83,Kh1983}. In this section, we find out the rough topology of the system $\mathcal{N}$, thus describing its rough Liouville equivalence class. Further, to describe the \textit{exact} topology of the system, i.e., to establish its Liouville equivalence class (for detailed definitions see \cite{igs}), we need to use general classifications of critical points and their neighborhoods \cite{LermIII,igs} and \textit{exact} topological invariants. By this term we mean the invariants that completely define Liouville foliations on some sufficient collection of 3-dimensional integral manifolds, e.g. Fomenko\,--\,Zieschang invariants \cite{FoZi1991} or marked loop molecules \cite{BolMet,BRF}. This will be done in the following sections.

According to Lemma~\ref{lem3}, the integral manifolds undergo topological transformations only in two cases, namely, when $(m,h)$ crosses the discriminant set of one of the product polynomials $\PR^2_i \QR^2_i$ ($i=1,2$) or when $(m,h)$ reaches the curve $L_0$.
In the first case, for $\ell \ne 0$, the map $\chi$ identifies in $\mathcal{N}$ two sets  $\tG(\ell,m,h)$ and $\tG(-\ell,m,h)$ different in $\mathcal{M}$. Therefore, in the sequel we by default suppose that $\ell \geqslant 0$. In the second case the set $\tG(0,m,h)$ is factorized by the action on $\Gamma_1{\times}\Gamma_2$  diagonal with respect to the induced actions on $\bbR^3(u_i,\QR_i,\PR_i)$
\begin{equation}\label{eq4_1}
\begin{array}{l}
  \chi_1: (u_1,\QR_1,\PR_1) \mapsto  (-u_1,\QR_1,-\PR_1), \\
  \chi_2: (u_2,\QR_2,\PR_2) \mapsto  (-u_2,-\QR_2,\PR_2).
\end{array}
\end{equation}

In what follows if $S$ stands for some set and $n$ is a positive whole number, then $n S$ denotes $n$ isolated copies of $S$.

\begin{figure}[!ht]
%figure02%
\centering
\includegraphics[width=0.8\textwidth, keepaspectratio = true]{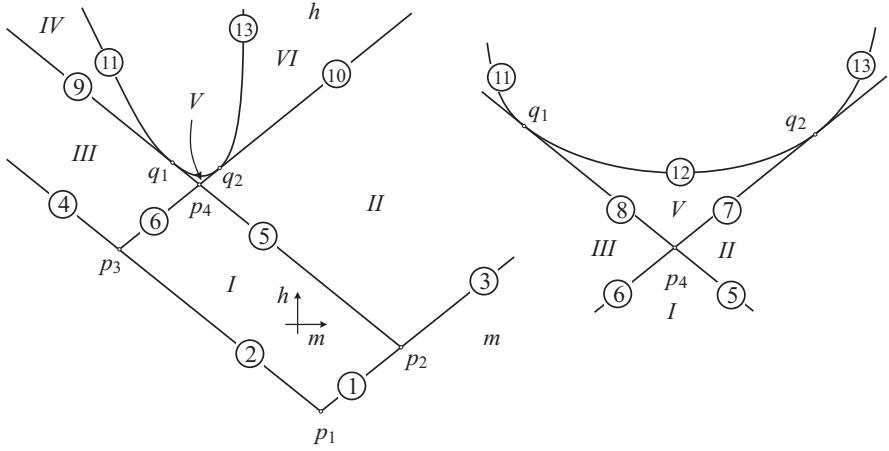}
\caption{Bifurcation diagram, chambers and segments.}\label{figbifset}
\end{figure}

\begin{theorem}\label{theo2}
The bifurcation diagram $\Sigma$ of the system $\mathcal{N}$ consists of the half-lines
\begin{equation*}
\begin{array}{lll}
R_b^-: & h=-r^2 m - 2b, & h \geqslant -a-b, \\
R_b^+: & h  =  - r^2 m + 2b, & h \geqslant -a+b, \\
R_a^-: & h  =  r^2 m - 2a, & h \geqslant -a-b, \\
R_a^+: & h  =  r^2 m + 2a,  & h \geqslant a-b
\end{array}
\end{equation*}
and of the curve $L_0$ defined by \eqref{eq3_2}.

The set $\Sigma$ divides the admissible region in the $(m,h)$-plane into six chambers $\ts{I}-\ts{VI}$ shown in Fig.~$\ref{figbifset}$. The corresponding regular integral manifolds are listed in the last column of Table~$1$. The integral surfaces in the pre-image of the $13$ segments forming the set $\Sigma$ are listed in the last column of Table~$2$.
\end{theorem}

\begin{table}
\begin{center}
\footnotesize
\renewcommand{\rul}{\rule[-4pt]{0pt}{12pt}}

\begin{tabular}{|c| c| c| c|c|c|c|}
\multicolumn{7}{l}{\fts{Table 1. Regular cases}}\\
\hline
\rul{\fns{Chamber}}
&
{\fns{Position of $\theta_{1,2}$}}
&
{\fns{$u_1\in$}}
&
{\fns{$u_2\in$}}
&
{$\Gamma_1$}
&
{$\Gamma_2$}
&
{$J_{m,h}$}
\\
\hline
\rul $\ts{I}_+$&$ - b < \theta _1  < b < a < \theta _2 $
&
\multirow{2}{*}{$[\,\tau_2,1 \,]$}
&
\multirow{2}{*}{$[\, - 1,\sigma_1 \,]$}
&
\multirow{2}{*}{$\Cir^1$}
&
\multirow{2}{*}{$\Cir^1$}
&
\multirow{2}{*}{$\Tor^2$}
\\
\hhline{|-|-|~|~|~|}
\rul $\ts{I}_-$&$\theta _2  <  - a <  - b < \theta _1  < b < a$&{}&{}&{}&{}&{}\\
\hline
\rul $\ts{II}_+$&$b < \theta _1  < a < \theta _2 $
&
\multirow{2}{*}{$[\,\tau_2,1 \,]$}
&
\multirow{2}{*}{$[\, - 1,1 \,]$}
&
\multirow{2}{*}{$\Cir^1$}
&
\multirow{2}{*}{$2\Cir^1$}
&
\multirow{2}{*}{$2\Tor^2$}
\\
\hhline{|-|-|~|~|~|}
\rul $\ts{II}_-$&$\theta _2  <  - a <  - b < b < \theta _1  < a$&{}&{}&{}&{}&{}\\
\hline
\rul $\ts{III}$&$ - a < \theta _2  <  - b < \theta _1  < b < a$
&
$[\,-1 ,1\,]$
&
$[\, - 1,\sigma _1 \,]$
&$2\Cir^1$&$\Cir^1$
&
{$2\Tor^2$}
\\
\hline
\rul $\ts{IV}$&$ - a <  - b < \theta _2  < \theta _1  < b < a$
&
$[\,-1,1\,]$
&
$[\,\sigma _2 ,\sigma _1 \,]$
&$2\Cir^1$&$2\Cir^1$
&
{$4\Tor^2$}
\\
\hline
\rul $\ts{V}$&$ - a < \theta _2  <  - b < b < \theta _1  < a$
&
$[\,-1 , 1\,]$
&
$[\, - 1,1\,]$
&$2\Cir^1$&$2\Cir^1$
&
{$4\Tor^2$}
\\
\hline
\rul $\ts{VI}_+$&$a < \theta _1  < \theta _2 $
&
\multirow{2}{*}{$[\,\tau_2,\tau_1\,]$}
&
\multirow{2}{*}{$[\, - 1,1 \,]$}
&
\multirow{2}{*}{$2\Cir^1$}
&
\multirow{2}{*}{$2\Cir^1$}
&
\multirow{2}{*}{$4\Tor^2$}
\\
\hhline{|-|-|~|~|~|}
\rul $\ts{VI}_-$&$\theta _2  <  - a < a < \theta _1 $&{}&{}&{}&{}&{}\\
\hline
\end{tabular}
\end{center}
\end{table}

The proof almost obviously follows from Lemma~\ref{lem3}. Let us make only some remarks.

The half-lines $R_{a}^\pm$ and $R_{b}^\pm$ correspond to the cases of a multiple root in one of the polynomials $\PR^2_i \QR^2_i$ ($i=1,2$). It happens if $\tau_{1,2}=\pm 1$ or $\sigma_{1,2}=\pm 1$. The corresponding critical motions are pendulum type motions
\begin{equation}\label{eq4_2}
\begin{array}{l}
R_a^{\pm}: \left\{ \begin{array}{l}
  \alpha_2 = \alpha_3 = 0, \quad \alpha_1 = \mp a,\quad \beta_1 = 0,\\
  \beta_2 = b \sin{\phi},\quad \beta_3 = b \cos{\phi},\\
  \omega_2 = \omega_3 = 0, \quad \omega_1 = \dot {\phi}, \quad 2\ddot{\phi} = b \cos {\phi},\\
  h=\dot{\phi}^2\pm a -b \sin\phi, \quad r^2m=\dot{\phi}^2 \mp a -b \sin\phi,
\end{array}\right. \\
R_b^{\pm}: \left\{ \begin{array}{l}
  \beta_1 = \beta_3 = 0,\quad \beta_2 = \mp b,\quad \alpha_2 = 0,\\
  \alpha_1 = a \cos{\phi},\quad \alpha_3 = a \sin{\phi},\\
  \omega_1 = \omega_3 = 0,\quad \omega_2 = \dot {\phi}, \quad 2\ddot{\phi} = -a \sin {\phi}, \\
  h=\dot{\phi}^2 \pm b -b \cos\phi, \quad r^2 m=-\dot{\phi}^2 \pm b +a \cos\phi.
\end{array}\right.
\end{array}
\end{equation}
They include closed orbits of rank 1, singular points of rank 0 at the intersections of the half-lines and separatrices of rank 1 of unstable singular points. From here we readily obtain the inequalities for $h$ on the half-lines.

Note that the mutual position of the values $\tau_i,\sigma_i,\pm 1$ is completely defined
by the position of $\theta_{1,2}=(\ell\mp 1)/(2m)$ with respect to $\pm a,\pm b$.
In Table~1, we present the inequalities for $\theta_{1,2}$ and, consequently, define the accessible regions (segments of oscillation) for $u_1,u_2$, the topology of $\Gamma_i$ and the resulting type of regular integral manifolds $J_{m,h}$. Here the $\pm$ sign attached to the notation of a chamber means the sign of $m$ in the corresponding part of this chamber. It affects the values $\theta_{1,2}$ but does not change the rest of information. The number of tori in $J_{m,h}$ for $\ell \ne 0$ equals $2^k$ where $k$ is the number of those radicals among $\PR_i,\QR_i$ which have constant sign on $\Gamma_1{\times}\Gamma_2$.

\begin{agre}\label{theagre1}
Suppose that the radical $\PR_i$ or $\QR_i$ does not vanish on the connected component of a regular integral manifold or a critical integral surface. Then we respectively denote
\begin{equation}\notag
  e_i =\sgn \PR_i \quad {\rm or} \quad d_i=\sgn \QR_i.
\end{equation}
\end{agre}

\begin{table}
\begin{center}
\footnotesize
\renewcommand{\rul}{\rule[-4pt]{0pt}{12pt}}
\begin{tabular}{|c| c| c| c|c|c|}
\multicolumn{6}{l}{\fts{Table 2. Critical cases}}\\
\hline
\rul{\fns{Seg.}}
&
{\fns{$u_1\in$}}
&
{\fns{$u_2\in$}}
&
{$\Gamma_1$}
&
{$\Gamma_2$}
&
{$J_{m,h}$}
\\
\hline
\rul $1$
&
$\{1_*\}$
&
$[\, - 1,\sigma_1 \,]$
&
{$\{\cdot\}$}
&
{$\Cir^1$}
&
{$\Cir^1$}
\\
\hline
\rul $2$
&
{$[\,\tau_2,1 \,]$}
&
{$\{ - 1_*\}$}
&
{$\Cir^1$}
&
{$\{\cdot\}$}
&
{$\Cir^1$}
\\
\hline
\rul $3$
&
$\{1_*\}$
&
$[\,-1,1\,]$
&$\{\cdot\}$
&$2\Cir^1$
&
{$2\Cir^1$}
\\
\hline
\rul $4$
&
$[\,-1 , 1\,]$
&
$\{ - 1_*\}$
&
$2\Cir^1$
&
$\{\cdot\}$
&
{$2\Cir^1$}
\\
\hline
\rul $5$
&
$[\,\tau_2 , 1\,]$
&
$[\,- 1,1_*\,]$
&
$\Cir^1$
&
$\mathcal{V}$
&
{$\mathcal{V}{\times}\Cir^1$}
\\
\hline
\rul $6$
&
$[\,-1_* , 1\,]$
&
$[\,- 1,\sigma_1\,]$
&
$\mathcal{V}$
&
$\Cir^1$
&
{$\mathcal{V}{\times}\Cir^1$}
\\
\hline
\rul $7$
&
$[\,-1_* , 1\,]$
&
$[\,- 1, 1\,]$
&
$\mathcal{V}$
&
$2\Cir^1$
&
{$2\mathcal{V}{\times}\Cir^1$}
\\
\hline
\rul $8$
&
$[\,-1 , 1\,]$
&
$[\,- 1, 1_*\,]$
&
$2\Cir^1$
&
$\mathcal{V}$
&
{$2\mathcal{V}{\times}\Cir^1$}
\\
\hline
\rul $9$
&
$[\,-1 , 1\,]$
&
$[\,- 1_*, \sigma_1\,]$
&
$2\Cir^1$
&
$\mathcal{V}$
&
{$2\mathcal{V}{\times}\Cir^1$}
\\
\hline
\rul $10$
&
$[\,\tau_2, 1_*\,]$
&
$[\,- 1, 1\,]$
&
$\mathcal{V}$
&
$2\Cir^1$
&
{$2\mathcal{V}{\times}\Cir^1$}
\\
\hline
\rul $11$
&
$[\,-1 , 1\,]$
&
$[\,- \sigma, \sigma\,]$
&
$2\Cir^1$
&
$2\Cir^1$
&
{$2\Tor^2$}
\\
\hline
\rul $12$
&
$[\,-1 , 1\,]$
&
$[\,- 1, 1\,]$
&
$2\Cir^1$
&
$2\Cir^1$
&
{$2\Tor^2$}
\\
\hline
\rul $13$
&
$[\,-\tau , \tau\,]$
&
$[\,- 1, 1\,]$
&
$2\Cir^1$
&
$2\Cir^1$
&
{$4\Tor^2$}
\\
\hline
\end{tabular}
\end{center}
\end{table}

In Table~2, we collect the information on the critical cases. In the first column we give the notation of the smooth segments of $\Sigma$ according to Fig.~\ref{figbifset}. In the corresponding accessible regions $\pm 1_*$ stand for the double root of $\PR_i^2\QR_i^2$.

More analysis is needed for the segments on the curve $L_0$. Here the roots of $\PR_i^2=f_i(u)$ become centrally symmetric and are denoted by $\pm \tau$, $\pm\sigma$ for $i=1,2$. The curve $L_0$ is tangent to $R_b^+$ and $R_a^+$ respectively at the points $q_{1,2}$
\begin{equation}\notag
\begin{array}{l}
  q_1 = \left( -\dfrac{1}{2b}, \dfrac{a^2+3b^2}{2b} \right),\; q_2 = \left( -\dfrac{1}{2a}, \dfrac{3a^2+b^2}{2a} \right).
\end{array}
\end{equation}
Let us formulate the result in a separate statement.

\begin{propos}\label{prop03}
The integral manifolds in the pre-image of the curve $L_0$ are as follows: $2\Tor^2$ on segments $11,12$; $4\Tor^2$ on segment $13$; $\mathcal{W}{\times}\Cir^1$ at the point $q_1$;
$2\mathcal{V}{\times}\Cir^1$ at the point $q_2$.
\end{propos}

\begin{proof} We see from Table~2 that the four components of the set $\Delta(0,m,h)$ differ by the signs of $\PR_1,\QR_2$ on segment $11$, by the signs of $\PR_1,\PR_2$ on segment $12$, and by the signs of $\QR_1,\PR_2$ on segment $13$. At the same time according to \eqref{eq4_1} the $\mathbb{Z}_2$-symmetry $\chi$ of $\Delta(0,m,h)$ simultaneously changes the signs of $\PR_1$ and $\QR_2$. This means that $\chi$ identifies the components having the same product $e_1 d_2$ on segment $11$ as in Fig.~\ref{figglu0},\,$a$. Here the arrows show the connected components of the direct product $\Gamma_1{\times}\Gamma_2$, the numbers stand for the resulting connected components in the quotient set. On segment $12$ (Fig.~\ref{figglu0},\,$b$) $\chi$ identifies the components with the opposite sign $e_1$ but the same sign $e_2$. Therefore, the result is $(\Gamma_1{\times}\Gamma_2)/\chi=2\Tor^2$. On segment $13$ the symmetry preserves all four components of $\Delta(0,m,h)$. The result is $(\Gamma_1{\times}\Gamma_2)/\chi=4\Tor^2$ (Fig.~\ref{figglu0},\,$c$).

\begin{figure}[!ht]
%figure03%
\centering
\includegraphics[width=0.88\textwidth, keepaspectratio = true]{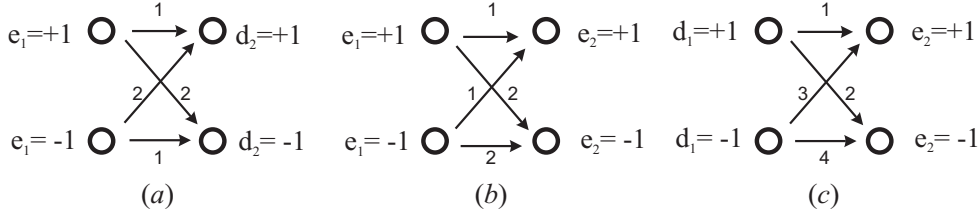}
\caption{Gluing the components in $J^{-1}(L_0)$.}\label{figglu0}
\end{figure}

Let us take the points $q_{1,2}$. The critical points in the pre-image are of rank~1 and form closed orbits. For the accessible regions we have
\begin{equation*}
  \begin{array}{llll}
     q_1: \quad u_1 \in [-1,1], & \tau > 1, &  u_2 \in [-1_*,1_*], & \sigma = 1;\\
     q_2: \quad u_1 \in [-1_*,1_*], & \tau = 1, &  u_2 \in [-1,1], & \sigma > 1.
   \end{array}
\end{equation*}
The sets $\Gamma_{1,2}$ for $q_1$ and $q_2$ are shown in Fig.~\ref{figpointsq1q2},\,$a$ and $b$. Again, $\chi$ acts as simultaneous central symmetry in $(u_1,\PR_1)$- and $(u_2,\QR_2)$-planes. At $q_1$, it glues together the components of $\Gamma_1$ and preserves $\Gamma_2$. The result is $\mathcal{W}{\times}\Cir^1$. At $q_2$, $\chi$ preserves the components of $\Gamma_2$ and therefore factorizes $\Gamma_1$. The result is $2\mathcal{V}{\times}\Cir^1$. This proves the statement.

\begin{figure}[!ht]
%figure04%
\centering
\includegraphics[width=\textwidth, keepaspectratio = true]{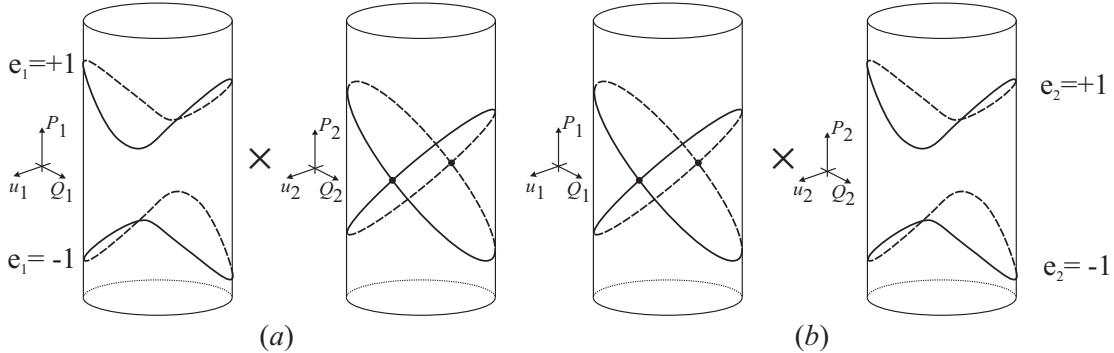}
\caption{The $\Delta$-sets at $q_1,q_2$.}\label{figpointsq1q2}
\end{figure}

\end{proof}

Now we describe all bifurcations in terms of the atoms according to the contemporary notation  \cite{Fo1988,FoZi1991,igs}. Let us recall some terminology from \cite{igs}.

Consider a 2-surface (two-dimensional compact manifold without boundary) and a Morse function $f$ on it having a critical value $f_0$. A 2-atom is a neighborhood of a connected component of the set $f_0-\delta \leqslant f \leqslant f_0+\delta$ for sufficiently small $\delta$, foliated into level lines of $f$ and considered up to the fiber equivalence. A 2-atom is supposed to have only one connected singular fiber $\{f=f_0\}$. If an atom $U$ is given, its singular fiber is denoted by $\mathcal{L}(U)$.

In Fig.~\ref{figatoms1}, the following atoms are shown: $A$ ($\partial A=\Cir^1$, $\mathcal{L}(A)=\{\cdot\}$); $B$~($\partial B=3\Cir^1$, $\mathcal{L}(B)=\mathcal{V}$);
$C_2$ ($\partial C_2=4\Cir^1$, $\mathcal{L}(C_2)=\mathcal{W}$); $C_1$ ($\partial C_1=2\Cir^1$, $\mathcal{L}(C_1)=\mathcal{W}$).
\begin{figure}[!ht]
%figure05%
\centering
\includegraphics[width=0.5\textwidth, keepaspectratio = true]{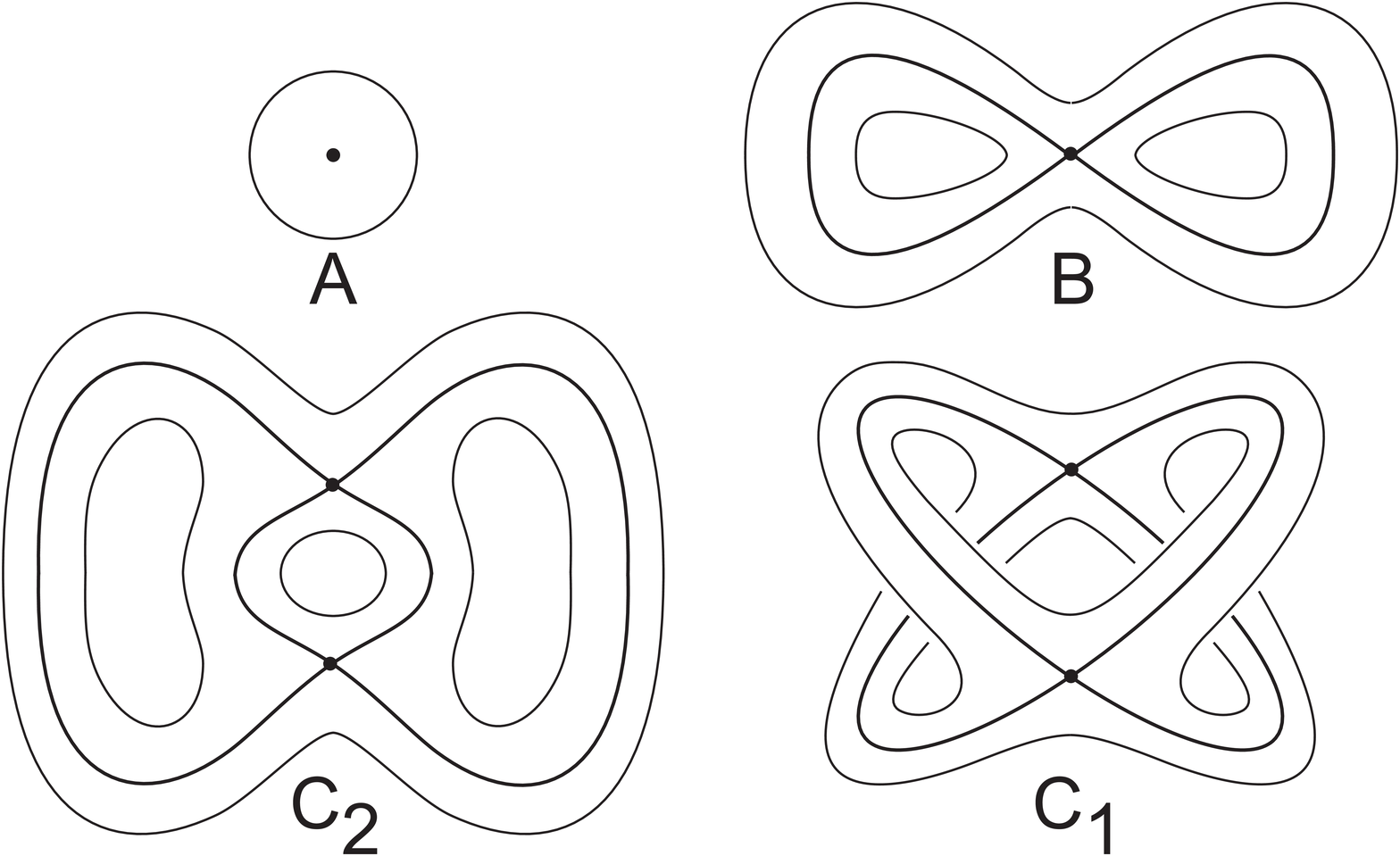}
\caption{Some known 2-atoms.}\label{figatoms1}
\end{figure}

For integrable Hamiltonian systems with two degrees of freedom, bifurcations of the Liouville tori are described in terms of 3-atoms. Let $F_1{\times}F_2$ be the integral map. Fixing a regular value $r_1$ of $F_1$ we obtain a 3-dimensional iso-$F_1$ manifold $\iso=\iso(r_1)$. Let us suppose for simplicity that $\iso$ is connected, otherwise we take one connected component of~$\iso$. Let $r_2$ be a critical value of $F_2$ on $\iso$ and $\mathcal{L}$ be the connected component of $(F_1{\times}F_2)^{-1}(r_1,r_2)$. Then $\mathcal{L}$ is called a singular leaf of the Liouville foliation on $\iso$. A 3-atom is a small enough connected neighborhood of $\mathcal{L}$ in $\iso$ containing no other singular leaves and invariant under the phase flow. In fact, 3-atoms as foliated manifolds are considered up to the fiber equivalence as defined in \cite{igs}. In the same way as before, for a given 3-atom $U$ we denote its singular leaf by $\mathcal{L}(U)$.

We easily obtain 3-atoms from the above mentioned 2-atoms by considering their direct products with a circle. Then, traditionally \cite{igs}, we keep for them the same notation. Thus, $\mathcal{L}(A)=\Cir^1$, $\mathcal{L}(B)=\mathcal{V}{\times}\Cir^1$,  $\mathcal{L}(C_2)=\mathcal{L}(C_1)=\mathcal{W}{\times}\Cir^1$ and the bifurcations of tori when crossing the singular leaves are
\begin{equation}\label{eq4_3}
  A: \varnothing \to \Tor^2, \quad B: \Tor^2 \to 2\Tor^2, \quad C_2: 2\Tor^2 \to 2\Tor^2, \quad C_1: \Tor^2 \to \Tor^2.
\end{equation}
Note that the 3-atom $C_1$ was predicted \cite{igs} but never has been met in real systems. Of course, bifurcations with non-symmetric atoms can be written the other way round depending on the direction in which we cross the bifurcation diagram.

In the case of minimal or maximal integral surfaces, symmetric atoms can be folded twice, so that the bifurcation $S\to U\to S$ turns into $\varnothing \to U \to 2S$. Let us use the notation $R$ for the atom of a minimal (maximal) torus. Here $\iso$ is a regular level of $F_1$ and on it $F_2$ has the form $(g-r_2)^2$ with regular function~$g$. The 2-atom $R$ (already not associated with any Morse function) is just an annulus foliated into circles, and the 3-atom $R$ is the direct product of an annulus and a circle foliated into 2-tori. Considering the torus $\{F_2=r_2\}$ as the singular leaf we have the following bifurcation $R: \varnothing \to 2\Tor^2$.

The atoms in \eqref{eq4_3} correspond to the critical points which are called simple \cite{LermIII} or non-degenerate \cite{igs}. For all points forming the motions \eqref{eq4_2} including those of rank 0 in the preimage of $p_1\ldots,p_4$ (see Fig.~\ref{figbifset}) the non-degeneracy is proved in \cite{RyKhMSb12}. The only exceptions are the motions in the pre-image of the points $q_1,q_2$, which are degenerate as critical points of rank~1 in the RS-system (with three degrees of freedom). Let us also mention that all tori in the pre-image of $L_0$ are degenerate as critical points of rank~2 in $\mathcal{P}$~\cite{RyKhMSb12}.

It is now easy to describe all non-degenerate bifurcations. We have the following atoms:

1) $A$ on segments $1,2$;

2) $2A$ on segments $3,4$;

3) $B$ on segments $5,6$;

4) $2B$ on segments $7,8,9,10$.

The points $p_i$ $(i=1,\ldots,4)$ have $h$-coordinate equal to $\mp a \mp b$ and are enumerated along the $h$-axis. The points $c_i$ of rank 0 in the pre-image of $p_i$ are also non-degenerate. For such points the local phase topology is described in terms of the almost direct products of 2-atoms \cite{Ng1995,igs}. Here we have only direct products. This fact follows immediately from the general classification of the cases with one singular point on a singular leaf \cite{LermIII,igs} and the atoms on the adjacent segments. Thus, small enough invariant under the phase flow neighborhoods $U_i$ of the points $c_i$ (called extended or saturated neighborhoods) are
\begin{equation}\label{eq4_4}
  U_1 = A{\times}A; \qquad U_{2,3} = A{\times}B; \qquad U_4 = B{\times}B.
\end{equation}

To finish the description of the rough topology we need to point out the bifurcations occurring on the pre-image of the curve $L_0$. Two of them are obvious. As shown in the proof of Proposition~\ref{prop03}, the symmetry $\chi$ on segments $11,12$ glues together pairwise the four tori of chambers $\ts{IV}$ and $\ts{V}$ respectively. Therefore we have here bifurcations with two atoms $R$.

\begin{figure}[!ht]
%figure06%
\centering
\begin{tabular}{m{0.75\textwidth} m{0.18\textwidth}}
\includegraphics[width=0.75\textwidth, keepaspectratio = true]{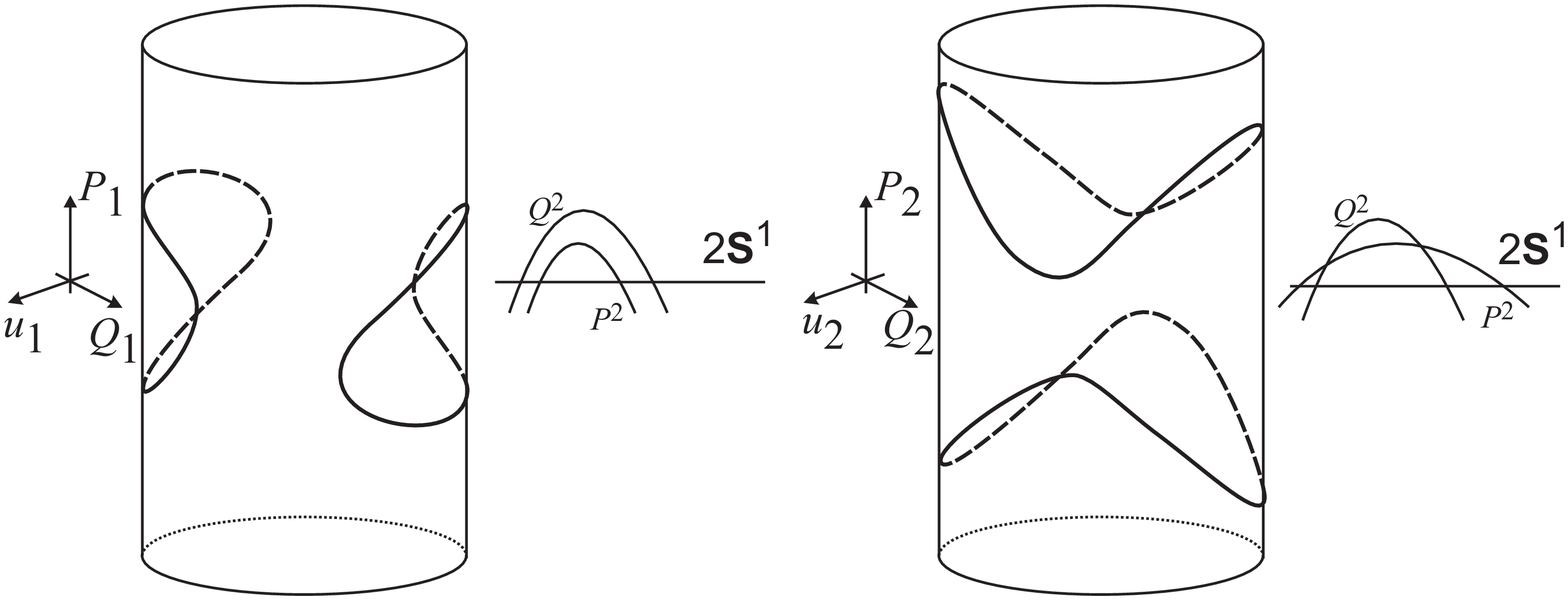}&
\includegraphics[width=0.18\textwidth, keepaspectratio = true]{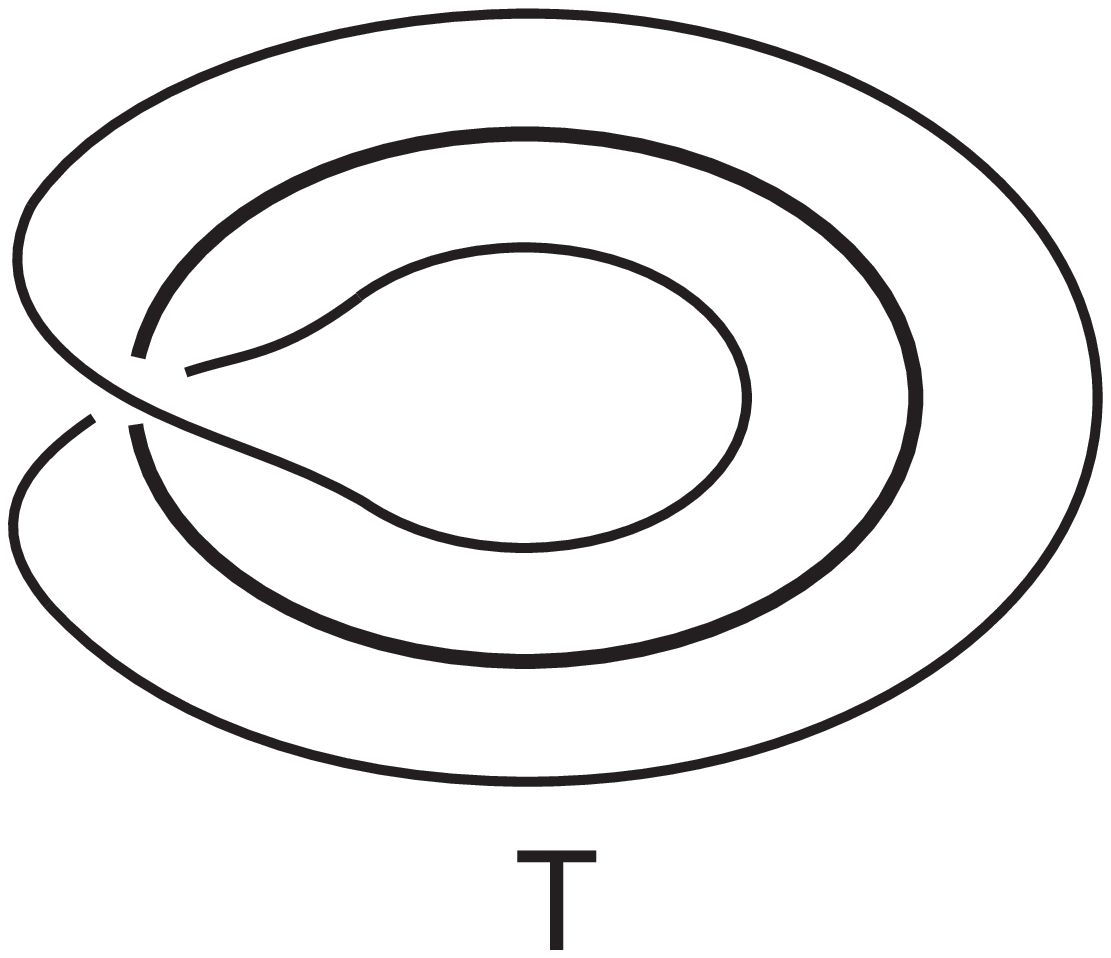}
\end{tabular}
\caption{The sets $\Gamma_i$ in chamber $\ts{VI}$ and the 2-atom $T$.}\label{figchamberVI}
\end{figure}

Let us consider a path reaching segment $13$ from chamber $\ts{VI}$. Each of the sets $\Delta(\pm \ell,m,h)\subset \mathcal{M}$ along this path has four components corresponding to the pair of signs $(d_1,e_2)$. Consider a continuous set of these components $\Tor^2(\ell;d_1,e_2)$ marked by $\ell \in (-\delta,\delta)$.
The component of $\Gamma_1(\ell)$ defined by $(d_1,e_2)$ is well projected onto the plane $(u_1,P_1)$ and the corresponding component of $\Gamma_2(\ell)$ is well projected onto the plane $(u_2,Q_2)$ as shown in Fig.~\ref{figchamberVI}. We then see that both $\chi_i: \Gamma_i(\ell) \to \Gamma_i(-\ell)$ are almost central symmetries on these planes and become real central symmetries on $\Gamma_i(0)$. Finally for a point on segment~$13$ we obtain the 3-atom $\Tor^2{\times}(-\delta,\delta)/\chi = M^2{\times}\Cir^1$, where $M^2$ is the M\"{o}bius band foliated into circles in the natural way with one singular central circle (the band's axis) twice shorter than all close ones. The product of the axis with a circle stands for the torus in the pre-image of $L_0$. The M\"{o}bius band itself foliated this way gives a non-orientable 2-atom, which we denote by $T$. The same notation we use for the corresponding 3-atom $M^2{\times}\Cir^1$. Its bifurcation in the direction from the border into the chamber is $T: \varnothing \to \Tor^2$, but unlike the atom $A$ having a circle as a singular leaf, here the singular leaf is a torus covered twice by the close regular torus as $\ell \to 0$. This 3-atom is impossible if we deal with a Hamiltonian system without degenerations of the symplectic structure.

To describe the topology in small enough neighborhoods of $q_1,q_2$ we first consider the situation arising in the covering manifold $\mathcal{M}$, where all transformations are easily seen from the sets $\Delta(\ell,m,h)$.

Consider a neighborhood of $q_1$ and unfold the picture from the integral constants space onto the plane $(\ell,m)$. We obviously obtain the cross formed by the lines $\ell=-2b m -1$ and $\ell=2b m +1$. On each of the lines the bifurcation is $2B$, along the horizontal line the bifurcation is $2C_2$ and along the vertical line the bifurcation is $2C_1$. In fact, this picture reflects two connected components of the neighborhood of the pre-image of $q_1$ in $\mathcal{M}$. On each component the bifurcation diagram together with the atoms is shown in Fig.~\ref{figbifQ1},\,$a$. The vertical iso-$M$ graphs change as in Fig.~\ref{figbifQ1},\,$b$, and the horizontal iso-$L$ graphs change as in Fig.~\ref{figbifQ1},\,$c$. This phenomenon can be called the splitting of the atom $C_2$. The possibility of topologically unstable systems is discussed in \cite{igs}. The case we obtain here is described as a possible transformation of iso-energy Fomenko graphs in \cite{RadRom}.

\begin{figure}[!ht]
%figure07%
\centering
\includegraphics[width=0.6\textwidth, keepaspectratio = true]{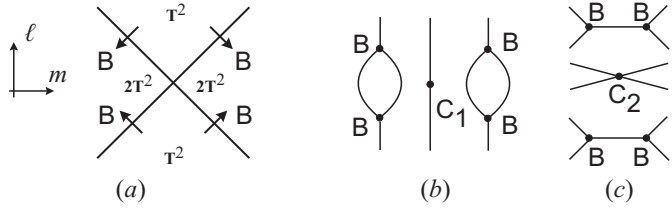}
\caption{The splitting of the atom $C_2$.}\label{figbifQ1}
\end{figure}

Let us consider a path reaching the point $q_1$ from chamber $\ts{III}$. All the way the curve $\Gamma_1$ has two components with different signs of $\PR_1$ while the closed curve $\Gamma_2(\ell,m,h)$ covers the whole set $\mathcal{W}$ as $\ell \to 0$ (Fig.~\ref{figchamberIII}).

\begin{figure}[!ht]
%figure08%
\centering
\includegraphics[width=0.8\textwidth, keepaspectratio = true]{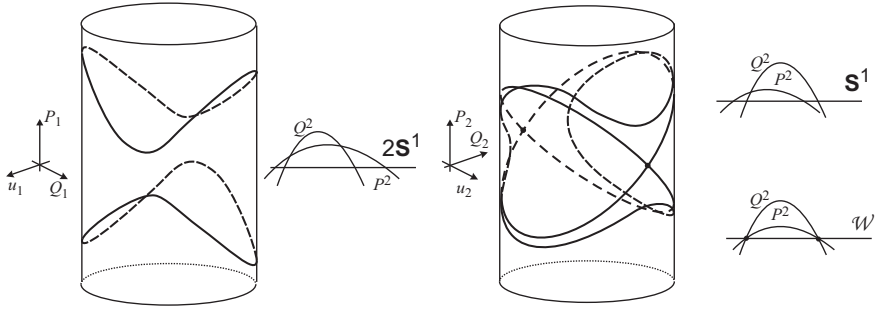}
\caption{The sets $\Gamma_i$ in chamber $\ts{III}$ while moving towards  $q_1$.}\label{figchamberIII}
\end{figure}

In $\mathcal{M}$, we unfold this path to a small vertical path $m=-\frac{1}{2b}, \ell \in (-\delta,\delta)$ and this process forms two atoms $\Cir^1{\times}C_1$, where $C_1$ denotes the 2-atom. Representing the sets $\Gamma_1 = 2\Cir^1$ and $\cup \Gamma_2(\ell) =  C_1$ on the plane as in Fig.~\ref{figAtomC1onplane}, we see that in this representation both $\chi_1$ and $\chi_2$ act as central symmetries, so $\chi_1$ identifies the components of $\Gamma_1 = 2\Cir^1$ and $\chi_2$ acts as a $\mathbb{Z}_2$-symmetry on $\cup \Gamma_2(\ell) = C_1$. Factorizing by the diagonal action, we can write, admitting some inexactness, $(2\Cir^1{\times}C_1)/\chi \cong (2\Cir^1/\chi_1){\times}C_1$. Therefore the pre-image of the chosen path in $\mathcal{N}$ gives one atom $\Cir^1{\times}C_1$, i.e., one 3-atom $C_1$. On the diagram $\Sigma$ of the system $\mathcal{N}$, crossing the point $q_1$ from the border into chamber $\ts{III}$ we have the bifurcation $C_1$ written in the form $\varnothing \to 2\Tor^2$.

\begin{figure}[!ht]
%figure09%
\centering
\includegraphics[width=0.3\textwidth, keepaspectratio = true]{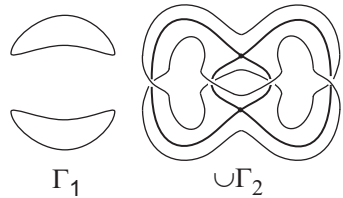}
\caption{The sets $\Gamma_1$ and $\cup \Gamma_2$ along the path in chamber $\ts{III}$.}\label{figAtomC1onplane}
\end{figure}

Let us consider a neighborhood of $q_2$ and, in the same way as in the previous case, unfold the picture from the integral constants space onto the plane $(\ell,m)$. We obtain the cross formed by the lines $\ell=-2a m -1$ and $\ell=2a m +1$. The topological picture in $\mathcal{M}$ again gives two copies of the bifurcation shown in Fig.~\ref{figbifQ1}. The essential difference appears after applying the factorization with respect to $\chi$.
Consider a small vertical path $m=-\frac{1}{2a}, \ell \in (-\delta,\delta)$ covering the path reaching $q_2$ transversally from chamber $\ts{II}$. The set $\Gamma_2 = 2\Cir^1$ does not transform in the whole neighborhood of $q_2$. Along the chosen path the union $\cup \Gamma_1(\ell)$ fills the 2-atom $C_1$ (see Fig.\ref{figchamberII}).

\begin{figure}[!ht]
%figure10%
\centering
\includegraphics[width=0.8\textwidth, keepaspectratio = true]{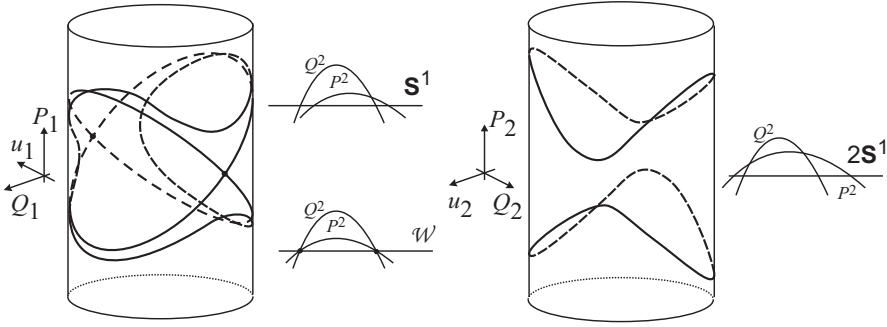}
\caption{The sets $\Gamma_i$ in chamber $\ts{II}$ while moving towards $q_2$.}\label{figchamberII}
\end{figure}

\begin{figure}[!ht]
%figure11%
\centering
\includegraphics[width=0.3\textwidth, keepaspectratio = true]{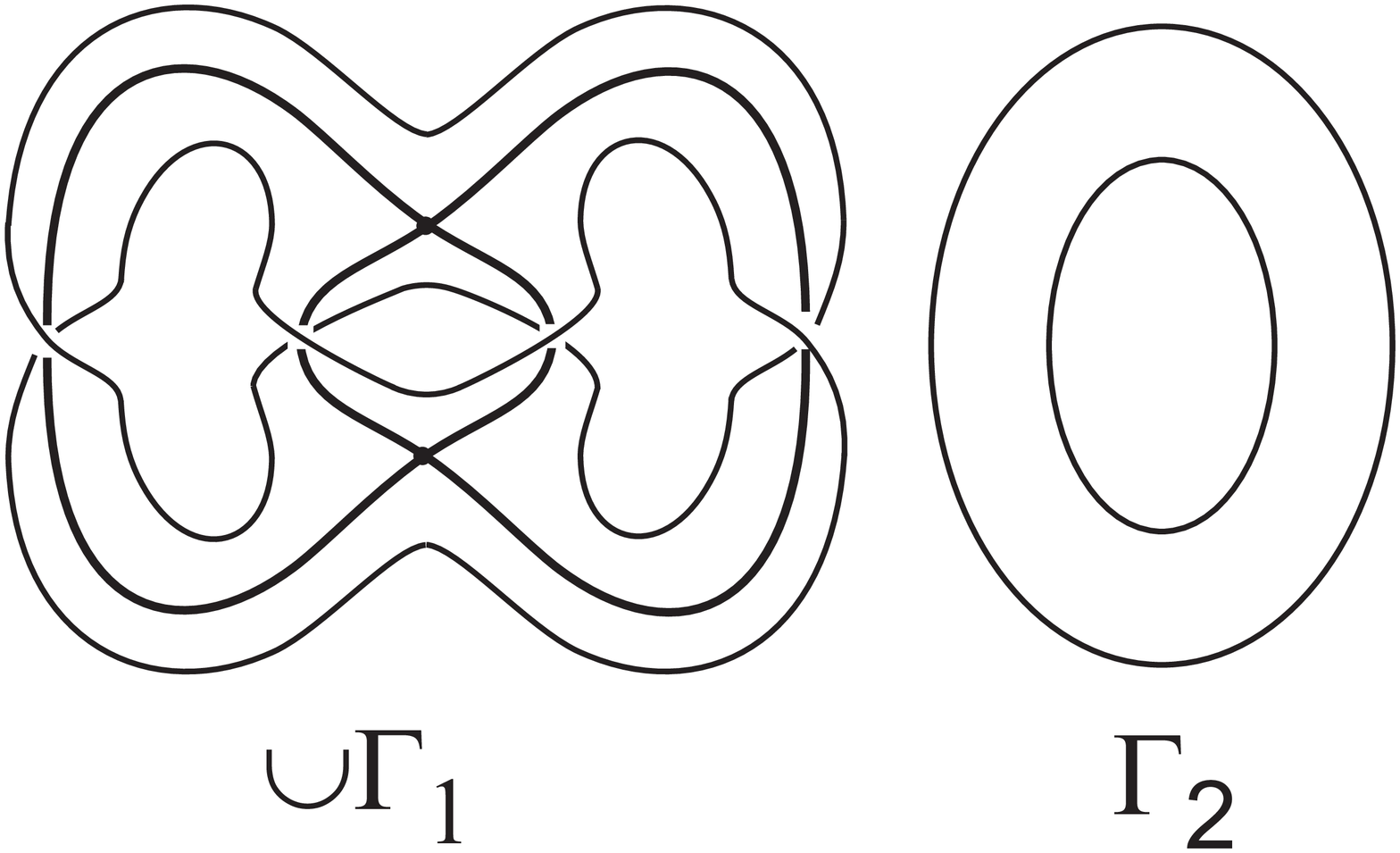}
\caption{The sets $\cup \Gamma_1$ and $\Gamma_2$ along the path in chamber $\ts{II}$.}\label{figAtomT1onplane}
\end{figure}

In this case we can show these sets in the plane as in Fig.~\ref{figAtomT1onplane}, where $\chi_1$ and $\chi_2$ act as central symmetries. Factorizing by the diagonal action, we can write $2(C_1{\times}\Cir^1)/\chi \cong 2(C_1/\chi_1){\times}\Cir^1$ and get the pre-image of the chosen path in $\mathcal{N}$ as a union of two atoms of the new type, which we denote by $T_1{\times}\Cir^1$. Here the 2-atom $T_1=C_1/\mathbb{Z}_2$ is shown in Fig.~\ref{figatomT1}. Again, we keep the same notation~$T_1$ for the corresponding 3-atom $T_1{\times}\Cir^1$. On the diagram $\Sigma$ of the system $\mathcal{N}$, crossing the point $q_2$ from the border into chamber $\ts{II}$ we have two simultaneous bifurcations $T_1: \varnothing \to \Tor^2$. The 3-manifold with boundary $T_1$ is non-orientable and is impossible in a system with non-degenerate symplectic structure.

\begin{figure}[!ht]
%figure12%
\centering
\includegraphics[width=0.35\textwidth, keepaspectratio = true]{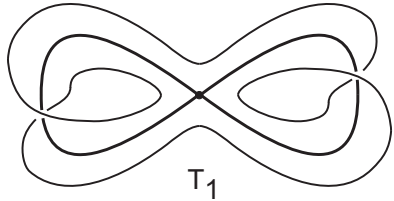}
\caption{The 2-atom $T_1$.}\label{figatomT1}
\end{figure}

Finally, we have obtained the topological description of the integral manifolds, critical integral surfaces and the bifurcations in $\mathcal{N}$ along any chosen path in the plane of the integral constants. This completes the rough topological analysis of the system $\mathcal{N}$.

\section{Exact topological analysis}\label{sec5}
The exact topological analysis is a way to establish the phase topology of the considered system up to Liouville equivalence \cite{igs}. In order to calculate main topological invariants of such equivalence, we need to describe the families of regular Liouville tori and to find the exact rules by which these families are glued to the boundaries of the bifurcation atoms. Let us recall some notions from the general theory \cite{igs}.

Given a Liouville integrable Hamiltonian (or almost Hamiltonian) system on a \mbox{4-d}i\-men\-sio\-nal manifold $M^4$, let us remove from the phase space all connected components of the integral surfaces containing critical points of the integral map $\mathcal{F}:M^4\to \bbR^2$, i.e., all singular leaves of the Liouville foliation. In our case, these leaves also include the whole level $\{L =0\}$. The connected component of the remaining set is called a \textit{family} of Liouville tori.

Consider a path $\gamma: [t_1,t_2] \to \bbR^2$ which is either closed $\gamma(t_1)=\gamma(t_2)$ or has its ends beyond the admissible region $\mathcal{F}(M^4)$. The pre-image $\iso_\gamma=\mathcal{F}^{-1}(\gamma([t_1,t_2]))$ is called a loop manifold. Under some simple transversality conditions it is indeed a smooth 3-dimensional manifold without boundary \cite{Os1991}. If the path $\gamma$ is a fixed level line of some first integral $\Phi$, we call $\iso_\gamma$ an iso-$\Phi$ manifold. Frequently, the role of $\Phi$ is played by the Hamiltonian $H$. Then $\iso_\gamma$ is called an iso-energy manifold. Identifying the points that belong to the same leaf of the Liouville foliation on $\iso_\gamma$ we obtain the rough Fomenko graph $W_\gamma$ with edges representing the families of regular tori and vertices corresponding to singular leaves of the foliation. Consider an edge of this graph bounded by two vertices. On the boundary tori of the atoms pointed by these vertices some pairs of coordinate cycles (bases of cycles) are defined called admissible coordinate systems (or admissible bases). Shift these bases to one regular torus corresponding to an inner point of the graph's edge. Two obtained bases are connected (in the one-dimensional homology group) with the so-called gluing matrix. It is an integer-valued matrix whose determinant is equal to $\pm 1$. In the orientable case without minimal or maximal tori the bases are chosen in such a way that the determinant is always equal to $-1$.
Endowing each edge of the graph with the gluing matrix we obtain one of the forms of the exact topological invariant of the Liouville foliation on $\iso_\gamma$. Usually, since gluing matrices are defined up to the changes of admissible coordinate systems, they are replaced by some sets of numerical invariants called marks. The resulting topological invariant is called the marked molecule and is denoted by $W^*_\gamma$ (see \cite{igs} for complete details). The goal of the exact topological analysis of an integrable system is to find the existing marked molecules and the corresponding loop manifolds for a reasonably full set of paths in the integral constants plane.

%\clearpage

\subsection{Families of tori and coordinate systems}
We return to the system $\mathcal{N}$ and use the advantages of the separation of variables to describe formally the families of regular tori and introduce, in some universal way, the coordinate system (the pair of independent cycles) on each family. In the general case of an integrable system with two degrees of freedom, each family is parameterized by the value of the integral map and the image of a family is an open connected set in $\bbR^2$. The image of one family can cover several chambers and the walls between them if there are walls on which some families do not bifurcate. For the system $\mathcal{N}$ this is not the case. Indeed, let us collect all information about the existing atoms and the number of regular tori in the chambers in Fig.~\ref{figbifwithatoms}. The arrows show the atoms and the direction in which the number of tori increases, the number of tori itself is given in squares. We see that on each wall all tori of the adjacent chambers are involved in the corresponding bifurcation. Indeed, in each chamber the accessible region for the separated variables in this system consists of one rectangle, therefore, all the tori projecting onto that rectangle bifurcate simultaneously.

\begin{figure}[!ht]
%figure13%
\centering
\includegraphics[width=0.8\textwidth, keepaspectratio = true]{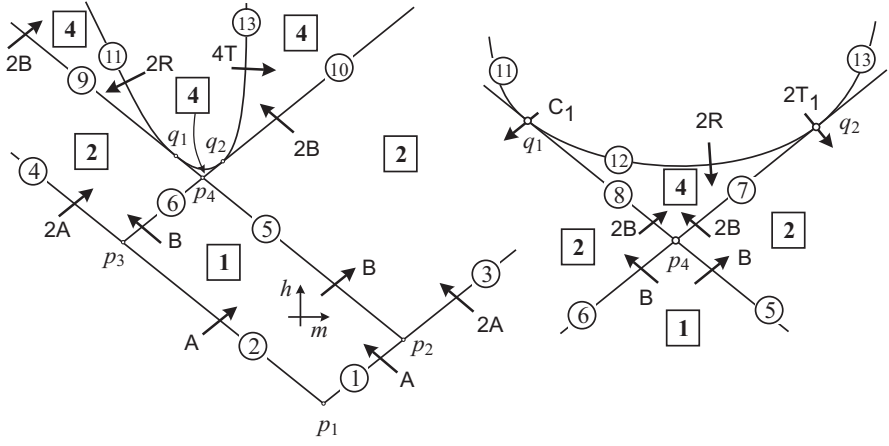}
\caption{Atoms and families.}\label{figbifwithatoms}
\end{figure}

Fix some chamber and consider the corresponding accessible regions
\begin{equation}\notag
  u_i \in [\xi_i,\eta_i]\qquad (i=1,2).
\end{equation}
Here $\xi_i,\eta_i$ are the roots of $\PR_i^2\QR_i^2=(1-u_i^2)f_i(u_i)$.
Theorem~\ref{theo1} states that regular tori are the connected components of $\Gamma_1{\times}\Gamma_2$, where $\Gamma_i$ are defined by \eqref{eq3_1}. Let us introduce the angular variables $\varphi_1,\varphi_2$ in such a way that
\begin{equation}\label{eq5_1}
\begin{array}{c}
  u_i = \xi_i \cos^2(\varphi_i)+\eta_i \sin^2(\varphi_i), \\[1mm]
  \sqrt{\eta_i - u_i}= \sqrt{\eta_i-\xi_i}\cos(\varphi_i), \qquad  \sqrt{u_i-\xi_i}= \sqrt{\eta_i-\xi_i}\sin(\varphi_i).
\end{array}
\end{equation}
Square roots of constant values are always supposed non-negative.
Substituting \eqref{eq5_1} into \eqref{eq2_4} we get the equations of $\Gamma_i$ containing $\sin \varphi_i,\cos \varphi_i $ and, maybe, some radicals which have constant sign on each connected component of the integral manifold. According to Agreement~\ref{theagre1} this sign is either $e_i$ or $d_i$ and, for the chosen connected component of $J_{m,h}$, it does not change inside the corresponding chamber. Finally we obtain that the family of tori is defined by the set of those signs out of $e_i,d_i$ which remain in the final expressions for $\PR_i,\QR_i$ on the connected components of $\Gamma_i$.

\begin{agre}\label{theagre2}
Having expressed $\PR_i,\QR_i$ on the connected component of $\Gamma_i$ in terms of $\sin\varphi_i,\cos\varphi_i,e_i,d_i$ we will always consider this component oriented by the direction of increasing of the angle $\varphi_i$.
\end{agre}

In fact, the separation of variables allows us to assign the universal orientation to the tori of each family. Indeed, the signs of the radicals in equations \eqref{eq2_5}, though arbitrary, are strictly consistent with the signs in equations \eqref{eq2_6}. Substituting \eqref{eq5_1} into \eqref{eq2_5}, we obtain
\begin{equation}\label{eq5_2}
    \dot \varphi_i = \varepsilon_i \Phi_i(\varphi_i) \qquad (i=1,2),
\end{equation}
where $\Phi_i > 0$ and $\varepsilon_i$ can equal $e_i$, $d_i$ or $\pm 1$. It means that the orientation of the cycle $\varepsilon_i\Gamma_i$ on all tori is given by the phase flow of the system $\mathcal{N}$.

\begin{agre}\label{theagre3}
We call the orientation of a regular torus positive if it is defined by the pair of cycles $(\varepsilon_1\Gamma_1,\varepsilon_2\Gamma_2)$, where $\varepsilon_i=\pm 1$ and the orientation of $\varepsilon_i\Gamma_i$ is induced by the phase flow.
\end{agre}

Since $\varepsilon_1,\varepsilon_2$ in \eqref{eq5_2} are the same for the whole family, by the described universal algorithm we fix a positive orientation of the family.

Consider an arbitrary atom $U$ together with some direction in which it is crossed along a chosen path $\gamma$ in $\bbR^2(m,h)$. The boundary $\partial U$ consists of regular tori divided in a natural way along $\gamma$ into incoming and outgoing ones. The admissible bases $(\lambda,\mu)$  on all boundary tori can be composed from the curves $\Gamma_i$ in one of the following ways $(\pm \Gamma_1, \pm\Gamma_2)$ or $(\pm \Gamma_2, \pm\Gamma_1)$ depending on the type of $U$. Some of the atoms may not have incoming or outgoing tori.

\begin{agre}\label{theagre4}
Choosing admissible bases on the boundary tori of an atom we always suppose that the orientation of these bases is positive (in the sense of Agreement~\ref{theagre3}) for all outgoing tori and negative for all incoming ones.
\end{agre}

Let us give the explicit formulas for the parameterized cycles $\Gamma_i$ and the positive bases on the families. Here we use the information given in Table~1.

For the curve $\Gamma_1$ we have three different cases. In chambers $\ts{I},\ts{II}$ the variable $u_1$ oscillates in $[\tau_2,1]$ and therefore we put
\begin{equation}\notag
\Gamma_1: \left\{
\begin{array}{l}
u_1=\tau_2 \cos^2\varphi_1+\sin^2\varphi_1,\\
\QR_1 = \sqrt{1-\tau_2}\sqrt{1+u_1(\varphi_1)}\cos \varphi_1,\\
\PR_1 = \sqrt{1-\tau_2}\sqrt{h_*[u_1(\varphi_1)-\tau_1]} \sin\varphi_1,\\
\sgn \dot \varphi _1 =+1.
\end{array}  \right.
\end{equation}

Let us once give remarks to the sign of $\dot\varphi_i$ applied in all similar representations of $\Gamma_i$. Making the formal substitution of the introduced expressions for $u_1,\QR_1,\PR_1$ into the first equation \eqref{eq2_5} we get
\begin{equation}\notag
  \dot \varphi_1 = \frac{1}{2}\sqrt{1+u_1(\varphi_1)}\sqrt{h_*[u_1(\varphi_1)-\tau_1]}.
\end{equation}
In all such cases the remaining non-constant square roots will have constant sign along the corresponding trajectory. We suppose them to be positive. Here there are two of them. To change both signs, it is enough to substitute $\varphi_1 \to \varphi_1+\pi$ without changing the sign of $\dot \varphi_1$. To change only one of the signs we may substitute $\varphi_1 \to -\varphi_1$ or $\varphi_1 \to \pi-\varphi_1$. This will change both the default orientation of $\Gamma_1$ and the sign of $\dot \varphi_1$ all over the chamber and, of course, will not affect the class of equivalent gluing matrices on the corresponding families.

Using the same method of the formal substitution with square roots considered positive, we obtain for chambers $\ts{III}-\ts{V}$
\begin{equation}\label{eq5_3}
\Gamma_1: \left\{
\begin{array}{l}
u_1=- \cos 2\varphi_1,\\
\QR_1 = \sin 2\varphi_1 ,\\
\PR_1 = e_1 \sqrt{h_*[u_1(\varphi_1)-\tau_1][u_1(\varphi_1)-\tau_2]},\\
\sgn \dot \varphi _1 = e_1,
\end{array}  \right.
\end{equation}
and for chamber $\ts{VI}$
\begin{equation}\label{eq5_4}
\Gamma_1: \left\{
\begin{array}{l}
u_1=\tau_2 \cos^2\varphi_1+\tau_1 \sin^2\varphi_1,\\
\QR_1 = d_1 \sqrt{1-u_1^2(\varphi_1)},\\
\PR_1 = \sqrt{-h_*}(\tau_1-\tau_2) \sin\varphi_1 \cos\varphi_1,\\
\sgn \dot \varphi _1 = d_1.
\end{array}  \right.
\end{equation}
Analogously, for the curve $\Gamma_2$ we obtain in chambers $\ts{I},\ts{III}$
\begin{equation}\label{eq5_5}
\Gamma_2: \left\{
\begin{array}{l}
u_2=-\cos^2\varphi_2+\sigma_1 \sin^2\varphi_2,\\
\QR_2 = \sqrt{1+\sigma_1}\sqrt{1-u_2(\varphi_2)}\sin \varphi_2,\\
\PR_2 = b \sqrt{1+\sigma_1}\sqrt{m[u_2(\varphi_2)-\sigma_2]} \cos\varphi_2,\\
\sgn \dot \varphi_2 =+1,
\end{array}  \right.
\end{equation}
in chambers $\ts{II},\ts{V},\ts{VI}$
\begin{equation}\label{eq5_6}
\Gamma_2: \left\{
\begin{array}{l}
u_2=- \cos 2\varphi_2,\\
\QR_2 = \sin 2\varphi_2 ,\\
\PR_2 = e_2 \, b\,\sqrt{m[u_2(\varphi_2)-\sigma_1][u_2(\varphi_2)-\sigma_2]},\\
\sgn \dot \varphi_2 = e_2,
\end{array}  \right.
\end{equation}
and in chamber $\ts{IV}$
\begin{equation}\label{eq5_7}
\Gamma_2: \left\{
\begin{array}{l}
u_2=\sigma_2 \cos^2\varphi_2+\sigma_1 \sin^2\varphi_2,\\
\QR_2 = d_2 \sqrt{1-u_2^2(\varphi_2)},\\
\PR_2 = b\sqrt{-m}(\sigma_1-\sigma_2) \sin\varphi_2 \cos\varphi_2,\\
\sgn \dot \varphi_2 = d_2.
\end{array}  \right.
\end{equation}

Finally we can say that for a given chamber the number of families is equal to the number of different combinations of $\sgn \dot \varphi_1,\sgn \dot \varphi_2$ and on each family the positive orientation is defined by the basis $(\sgn \dot \varphi_1 \Gamma_1,\sgn \dot \varphi_2 \Gamma_2)$. The complete information on the families and the corresponding bases is given in Table~3.

\begin{center}
\footnotesize
\renewcommand{\rul}{\rule[-4pt]{0pt}{12pt}}
\begin{tabular}{|c| c| c|c|c|c|c|}
\multicolumn{7}{l}{\fts{Table 3. Families and bases for the chambers}}\\
\hline
\rul{\fns{Chamber}}
&
$\ts{I}$
&
$\ts{II}$
&
$\ts{III}$
&
$\ts{IV}$
&
$\ts{V}$
&
$\ts{VI}$\\
\hline
{\renewcommand{\arraystretch}{0.8}\fns{\begin{tabular}{c}\strut Number of\\ families\end{tabular}}}
&
$1$
&
$2$
&
$2$
&
$4$
&
$4$
&
$4$\\
\hline
\rule[-10pt]{0pt}{26pt}{\fns{Bases}}
&
$\colu{\Gamma_1}{\Gamma_2}$
&
$\colu{\Gamma_1}{e_2 \Gamma_2}$
&
$\colu{e_1 \Gamma_1}{\Gamma_2}$
&
$\colu{e_1 \Gamma_1}{d_2 \Gamma_2}$
&
$\colu{e_1 \Gamma_1}{ e_2 \Gamma_2}$
&
$\colu{d_1 \Gamma_1}{e_2 \Gamma_2}$\\
\hline
\end{tabular}
\end{center}

Let us denote by $a_i$ the atoms arising in the pre-image of segments $1,\ldots,13$ of the bifurcation diagram, where $i$ is the number assigned to the corresponding segment. If there are two atoms in the pre-image of a point on segment $i$, we denote them by $a_i^{\varepsilon_j}$ using for $\varepsilon_j=\pm 1$ the sign built out of the fixed signs of the radicals as defined above. The only exception is $a_{13}$ having four components. To each of them we assign the pair of sings. Obviously, in this notation we have the following atoms
\begin{equation}\notag
  a_1, \; a_2, \; a_3^{e_2}, \; a_4^{e_1}, \; a_5, \; a_6, a_7^{e_2}, \; a_8^{e_1}, a_9^{e_1},\; a_{10}^{e_2},\; a_{11}^{e_1 d_2}, \; a_{12}^{e_2}, \; a_{13}^{(e_2,d_1)}.
\end{equation}
For a non-degenerate 3-atom, we can say that its sign is defined by the sign of the connected component of the curve out of $\Gamma_1,\Gamma_2$ that does not bifurcate at this moment. The signs of the atoms in the pre-image  of a point on the curve $L_0$ (segments $11$, $12$, and $13$) separate those tori families which are not identified with each other upon reaching $L_0$.

To establish for each atom (excluding for a while the points of $L_0$) the uniquely defined admissible cycles ($\mu$-cycles for $A$ and $\lambda$-cycles for hyperbolic atoms), we use Table~2 and equations \eqref{eq5_3}--\eqref{eq5_7} giving the orientation on $\Gamma_i$ induced by the phase flow. The general rule of choosing this cycle is as follows. First, we take the curve $\Gamma_i$ for which the accessible region does not contain a double root (i.e., $\pm 1_*$). Then we multiply it by the sign of $\dot{\varphi}_i$ from the adjacent chamber. This sign, of course, coincides with the sign of the corresponding cycle in the basis taken for the positive orientation of the family in this chamber in Table~3. For example, on segment $4$ with two atoms $A$ we take $\mu=\pm \Gamma_1$ from Table~2 and choose $e_1$ for the sign from the basis of chamber $\ts{III}$ in Table~3. On segment $10$ with two atoms $B$ we take $\lambda=\pm \Gamma_2$ from Table~2 and then choose $e_2$ for the sign from the basis of any of chambers $\ts{II}$ or $\ts{VI}$ in Table~3.
The second cycle in the pair $(\lambda,\mu)$ defining an admissible coordinate system is chosen according to Agreements~\ref{theagre3} and \ref{theagre4}.

It is convenient to collect in one table all admissible coordinate systems for the atoms on the segments of the bifurcation diagram written out for some globally fixed direction of crossing these atoms. Let us take the direction of the increasing $h$-coordinate and denote the corresponding pairs of cycles for an atom $a_i$ by $\mathcal{B}_i\inn$ and $\mathcal{B}_i\out$.
While constructing the loop molecules we may meet with the necessity to cross some atoms in the inverse direction. Then we denote such admissible bases for an atom $a_i$ by $\mathcal{C}_i\inn$ and $\mathcal{C}_i\out$. The connection between these bases is obvious
\begin{equation}\label{eq5_8}
  \mathcal{C}_i\inn = \matr{1}{0}{0}{-1}\mathcal{B}_i\out,\qquad \mathcal{C}_i\out = \matr{1}{0}{0}{-1}\mathcal{B}_i\inn.
\end{equation}

\begin{table}
\centering
\footnotesize
\renewcommand{\rul}{\rule[-10pt]{0pt}{26pt}}
\tabcolsep=0.35em
\noindent\begin{tabular}{|c| c| c|c| c| c|c| c|}
\multicolumn{8}{l}{\fts{Table 4. Admissible bases for the atoms}}\\
\hline
\rul{\fns{Seg.}}
&
1
&
2
&
3
&
4
&
5
&
6
&
7
\\
\hline
\rul{\fns{$\mathcal{B}\inn$}}
&
$-$
&
$-$
&
$-$
&
$-$
&
$\colu{\Gamma_1}{-\Gamma_2}$
&
$\colu{\Gamma_2}{\Gamma_1}$
&
$\colu{e_2\Gamma_2}{\Gamma_1}$
\\
\hline
\rul{\fns{$\mathcal{B}\out$}}
&
$\colu{\Gamma_1}{\Gamma_2}$
&
$\colu{-\Gamma_2}{\Gamma_1}$
&
$\colu{-\Gamma_1}{-e_2\Gamma_2}$
&
$\colu{\Gamma_2}{-e_1\Gamma_1}$
&
$\colu{\Gamma_1}{e_2\Gamma_2}$
&
$\colu{\Gamma_2}{-e_1\Gamma_1}$
&
$\colu{e_2\Gamma_2}{-e_1\Gamma_1}$
\\
\hhline{|=|=|=|=|=|=|=|=|}
\rul{\fns{Seg.}}
&
8
&
9
&
10
&
11
&
$12_1$
&
$12_2$
&
13
\\
\hline
\rul{\fns{$\mathcal{B}\inn$}}
&
$\colu{e_1\Gamma_1}{-\Gamma_2}$
&
$\colu{e_1\Gamma_1}{-\Gamma_2}$
&
$\colu{e_2\Gamma_2}{\Gamma_1}$
&
$\colu{e_1\Gamma_1}{-d_2\Gamma_2}$
&
$\colu{e_1\Gamma_1}{-e_2\Gamma_2}$
&
$\colu{e_2\Gamma_2}{e_1\Gamma_1}$
&
$\colu{e_2\Gamma_2}{d_1(\Gamma_1{+}\Gamma_2)}$
\\
\hline
\rul{\fns{$\mathcal{B}\out$}}
&
$\colu{e_1\Gamma_1}{e_2\Gamma_2}$
&
$\colu{e_1\Gamma_1}{d_2\Gamma_2}$
&
$\colu{e_2\Gamma_2}{-d_1\Gamma_1}$
&
$-$
&
$-$
&
$-$
&
$-$
\\
\hline
\end{tabular}
\end{table}

The result is given in Table~4. For the curve $L_0$ the general rule does not work since the tori in the pre-image are in fact regular manifolds. For the new atoms $R$, $T$, $C_1$ and $T_1$ we need to establish some other rules.

For the atoms $R$ we take admissible bases with the $\lambda$-cycle coming from the adjacent atom $B$ in the molecule. Thus, the notation $12_1$ stands for the case of the edge coming from segment $8$ (the part of the loop molecule of the point $q_1$) and $12_2$ is used for the bases obtained along the edge coming from segment $7$ (the part of the loop molecule of the point $q_2$).

Consider the 3-atom $T=M^2{\times}\Cir^1$ on segment $13$. Obviously, $\Cir^1$ here stands for the global $\lambda$-cycle $e_2\Gamma_2$ which came from segment $10$. So let it be the first cycle of the admissible basis on $\partial T = \Tor^2$. For the $\mu$-cycle let us take the circle in the M\"{o}bius band that covers twice the middle line of the band oriented according to Agreement~\ref{theagre4}. On the family in chamber $\ts{VI}$ corresponding to the signs $(e_2,d_1)$, the basis $\mathcal{B}_{13}\inn$ is negatively oriented. Then from Table~3, its orientation coincides with that of $(d_1\Gamma_1,-e_2\Gamma_2)$, which is the same as of $(e_2\Gamma_2,d_1\Gamma_1)$.

Let us write out equations \eqref{eq5_4} and \eqref{eq5_6} on the curve $L_0$ ($0<\tau<1, \sigma>1$):
\begin{equation}\notag
\Gamma_1: \left\{
\begin{array}{l}
u_1=- \tau \cos 2\varphi_1,\\
\QR_1 = d_1 \sqrt{1-u_1^2(\varphi_1)},\\
\PR_1 = \sqrt{-h_*}\tau \sin 2\varphi_1,\\
\sgn \dot \varphi _1 = d_1,
\end{array}  \right. \quad
\Gamma_2: \left\{
\begin{array}{l}
u_2=- \cos 2\varphi_2,\\
\QR_2 = \sin 2\varphi_2 ,\\
\PR_2 = e_2 \, b\,\sqrt{-m[\sigma^2-u_2^2(\varphi_2)]},\\
\sgn \dot \varphi_2 = e_2.
\end{array}  \right.
\end{equation}
It is easy to see that $\chi$ acts on $\Gamma_1{\times}\Gamma_2$ as the simultaneous shift $\varphi_1\to \varphi_1+\frac{\pi}{2}$ and $\varphi_2\to \varphi_2+\frac{\pi}{2}$. Therefore the circle folding twice is $\Gamma_1+\Gamma_2$ and we come to the basis for segment $13$ as in Table~4.

Since the 3-atom $T_1$ is non-orientable and $\partial T_1$ consists of only one torus, the solution here is standard, i.e., for an outgoing torus we take the positive orientation of the family, but for an incoming one we take the negative orientation of the family. On the contrary, the 3-atom $C_1$ is orientable and the choice of its orientation defines the orientations on two boundary tori. In our case at the point $q_1$ both of them are incoming (supposing that one of the coordinates $h$ or $m$ increases while another is fixed). Two families come to the point $q_1$ from chamber $\ts{II}$. Choosing the negative orientation from the families as given by the pair $(\Gamma_1,-e_2\Gamma_2)$ (see Table~3), we have one of them different from the atom's orientation. We mark this situation by the notation $C_1^*$. Note that we may choose the orientation on $\partial C_1$ as $(\Gamma_1,\Gamma_2)$. Then it is consistent with some orientation of $C_1$. Obviously, in this case one of the gluing matrices on the incoming edges will have the determinant equal to $+1$. If the topology of the molecule makes it possible to change the direction of this edge and the following ones without general contradiction, then we can change orientations of the corresponding families of tori and obtain an orientable molecule. But if this edge is a part of a loop, such a change may be impossible. Then the molecule defines a non-orientable loop manifold.

\subsection{Loop molecules of non-degenerate singular points}
As it was mentioned above, in the system $\mathcal{N}$ we have exactly four critical points of rank 0, all four are non-degenerate and their neighborhoods have representations \eqref{eq4_4}. Therefore the corresponding loop molecules are well-known and completely classified by the Bolsinov theorem \cite{igs}. Nevertheless, in this section we demonstrate the use of separated variables for the process of constructing these loop molecules.

Let us take small closed paths around the points $p_1,\ldots,p_4$ in the clockwise direction. Then from Tables~3 and 4 we have the following chains:
\begin{equation}\label{eq5_9}
\begin{array}{l}
  p_1: \;(a_2)\;\mathcal{B}_2\out \to \mathcal{C}_1\inn \;(a_1); \\
  p_2: \;(a_4)\;\mathcal{B}_4\out \to \mathcal{C}_6\inn \;(a_6)\; \mathcal{C}_6\out \to \mathcal{C}_2\inn \;(a_2);\\
  p_3: \;(a_1)\;\mathcal{B}_1\out \to \mathcal{B}_5\inn \;(a_5)\; \mathcal{B}_5\out \to \mathcal{C}_3\inn \;(a_3);\\
  p_4: \;(a_5)\;\mathcal{C}_5\out \to \mathcal{B}_6\inn \;(a_6)\; \mathcal{B}_6\out \to \mathcal{B}_8\inn \;(a_8)\; \mathcal{B}_8\out \to \\
  \phantom{p_4: \;} \to \mathcal{C}_7\inn \;(a_7)\; \mathcal{C}_7\out \to \mathcal{C}_5\inn \;(a_5).
\end{array}
\end{equation}
Substituting the bases from Table~4 and \eqref{eq5_8} we readily obtain not only the gluing matrices but also the rules by which families choose their bounding atoms. Indeed, let us take for example the edge connecting $a_8^{e_1}$ and $a_7^{e_2}$. We have
\begin{equation*}
  \mathcal{B}_8\out = \colu{e_1\Gamma_1}{e_2\Gamma_2}, \qquad \mathcal{C}_7\inn = \colu{e_2\Gamma_2}{e_1\Gamma_1}=\matr{0}{1}{1}{0} \mathcal{B}_8\out.
\end{equation*}
We see that the two families started from the $B$-atom $a_8^+$ ($e_1=+1$) differ by the sign $e_2$ and therefore come to different $B$-atoms on segment $7$.

For the sake of brevity we denote some $(2{\times}2)$-matrices needed below
\begin{equation}\notag
\begin{array}{c}
  E_\pm=\matr{\pm 1}{0}{0}{\mp 1}, \quad D_\pm=\matr{0}{\pm 1}{\pm 1}{0}, \quad C_\pm=\matr{1}{0}{\pm 1}{-1}.
\end{array}
\end{equation}
As usual, $E$ stands for the identity matrix. The loop molecules of non-degenerate singularities $c_i$ ($J(c_i)=p_i$) generated by the chains \eqref{eq5_9} are shown in Fig.~\ref{figmols1234}: ($a$) for $c_1$, ($b$) for $c_2,c_3$, and ($c$) for $c_4$.

\begin{figure}[!ht]
%figure14%
\centering
\includegraphics[width=0.75\textwidth, keepaspectratio = true]{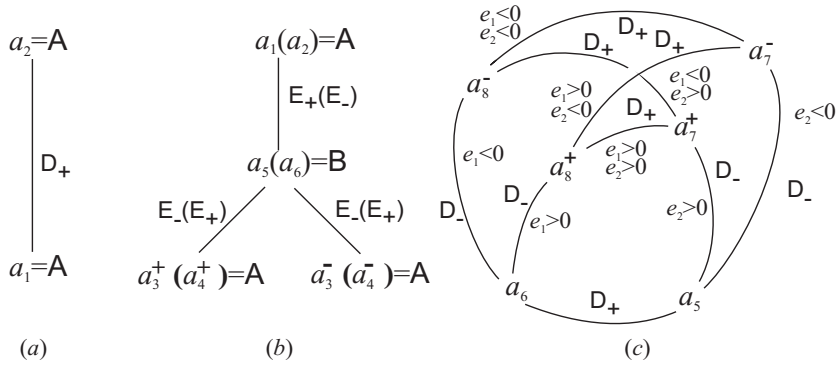}
\caption{Loop molecules of non-degenerate points.}\label{figmols1234}
\end{figure}

Of course, our method of global choice of orientation affects the gluing matrices, but since the possible changes (e.g. changing the order of the separation variables) are applied to all tori simultaneously, the gluing matrices of the molecules considered in this section simultaneously change their signs ($E_\pm \to E_\mp, D_\pm \to D_\mp$). This only leads to another representatives of the same exact topological invariants.

\subsection{Loop molecules of degenerate closed orbits}
In the system $\mathcal{N}$ the set of degenerate closed orbits consists of the motions \eqref{eq4_2} in the pre-image of the points $q_1,q_2$.

\begin{theorem}\label{theo3}
The loop molecules of the points $q_1,q_2$ can be represented in the form shown in Fig.~\ref{figlmolQ1Q2}. The first one is connected, while the second consists of two equivalent connected components differing by the sign $e_2$.
\end{theorem}

\begin{figure}[!ht]
%figure15%
\centering
\includegraphics[width=\textwidth, keepaspectratio = true]{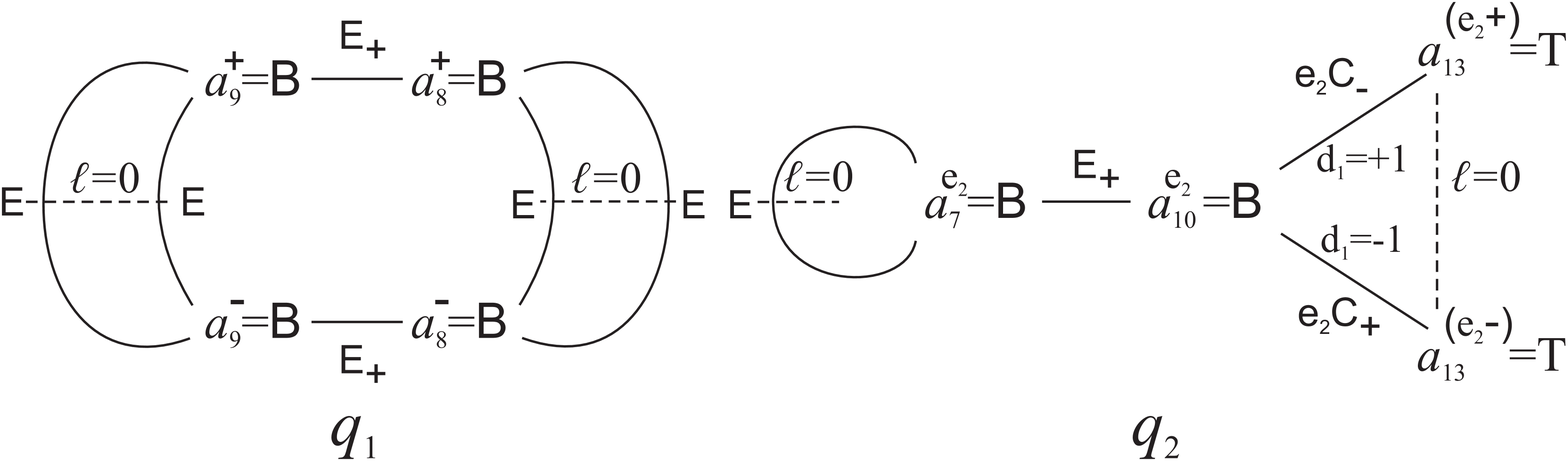}
\caption{Loop molecules of the points $q_1,q_2$.}\label{figlmolQ1Q2}
\end{figure}

\begin{proof}
Consider a closed path surrounding the point $q_1$ clockwise. Recall that the atoms $B$ both in $a_8$ and $a_9$ differ by the sign $e_1$. At the same time on $L_0$ we have to identify the families from chamber $\ts{IV}$ having the same product $e_1d_2$ and from chamber $\ts{V}$ having the same sign $e_2$. Then we obtain the picture of the molecular edges printed over a piece of the bifurcation diagram as shown in Fig.~\ref{figtomolQ1xn}.

\begin{figure}[!ht]
%figure16%
\centering
\includegraphics[width=0.55\textwidth, keepaspectratio = true]{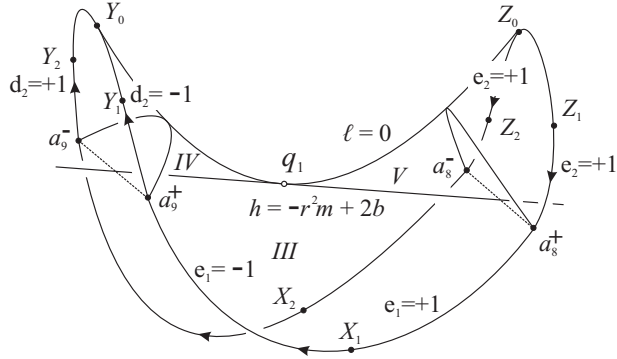}
\caption{Constructing the loop molecule of the point $q_1$.}\label{figtomolQ1xn}
\end{figure}

In the whole neighborhood of $q_1$ all $\lambda$-cycles on the regular tori are induced by the periodic critical trajectories of the atoms $a_8,a_9$. Then from Table~2 we readily obtain that both $\lambda_8$ and $\lambda_9$ (no matter, incoming or outgoing) are defined as $e_1\Gamma_1$.
Taking the chains marked with the points $Z_i,X_i,Y_i$ we get
\begin{equation}\notag
\begin{array}{ll}
(Z_i)\;\mathcal{C}_{12}\out \to \mathcal{C}_8\inn \;(a_8)\; \mathcal{C}_8\out \to \;(X_i)\; \to \mathcal{B}_9\inn \;(a_9)\; \mathcal{B}_9\out \to \mathcal{B}_{11}\inn \;(Y_i).
\end{array}
\end{equation}
From Table~4, the gluing matrix is $E_+$ at both points $X_i$.

Let us find the gluing matrices on the arcs connecting the atoms $a_9^{\pm}$. To be definite, let us take the points $Y_{1,2}$ on the arc with $e_1d_2=+1$ and shift the corresponding bases $\mathcal{B}_9\out$ to the torus $J^{-1}(Y_0)$. We get two bases
$(\Gamma_1,\Gamma_2)$ and $(-\Gamma_1,-\Gamma_2)$.
On segment $11$ with $|\tau_i|=\tau>1$ we obtain from \eqref{eq5_3} and \eqref{eq4_1}
\begin{equation}\notag
\chi_1(\Gamma_1): \left\{
\begin{array}{l}
u_1= \cos 2\varphi_1,\\
\QR_1 = \sin 2\varphi_1 ,\\
\PR_1 = - e_1 P ,\\
\sgn \dot \varphi _1 = e_1.
\end{array}  \right. \quad  -\Gamma_1: \left\{
\begin{array}{l}
u_1=- \cos 2\varphi_1,\\
\QR_1 = \sin 2\varphi_1 ,\\
\PR_1 = - e_1 P ,\\
\sgn \dot \varphi _1 = -  e_1,
\end{array}  \right.
\end{equation}
where $P=\sqrt{-h_*[\tau^2 - \cos^2 2\varphi_1]}>0$. The obvious substitution $\varphi_1 \to \frac{\pi}{2}-\varphi_1$ turns $\chi_1(\Gamma_1)$ into $-\Gamma_1$. A similar reasoning leads to the equality $\chi_2(\Gamma_2)=-\Gamma_2$ on $J^{-1}(Y_0)$. Finally we have that the map $\chi$ identifies two incoming bases at $Y_0$, so they give the same basis on the image torus in $\mathcal{N}$. Therefore the gluing matrix on the arc is the identity matrix $E$. Another arc $e_1 d_2 = -1$ is considered analogously and gives the same result.

Now let us turn to the arcs in chamber $\ts{V}$. Shifting the bases $\mathcal{C}_8\inn$ from the points $Z_{1,2}$ to the point $Z_0$ (i.e., to the torus $J^{-1}(Z_0)$) we get two bases
$(\Gamma_1,-\Gamma_2)$ and $(-\Gamma_1,-\Gamma_2)$.
In the same way as above, \eqref{eq5_3} yields $\chi_1(\Gamma_1)=-\Gamma_1$, but from \eqref{eq5_6} we obviously have $\chi_2(\Gamma_2)=\Gamma_2$. So again $\chi$ identifies two bases at $Z_0$ and they give the same basis on the image torus in $\mathcal{N}$. The gluing matrix is $E$. This proves the statement on the loop molecule of $q_1$.

Consider a closed path surrounding the point $q_2$ clockwise. Recall that the atoms $B$ both in $a_7$ and $a_{10}$ differ by the sign $e_2$. At the same time on $L_0$ we have to identify the families from chamber $\ts{V}$ having the same sign $e_2$. The families in chamber $\ts{VI}$ are not identified with each other and end with the atoms $T$. This leads to two rough invariants homeomorphic to that shown in Fig.~\ref{figlmolQ1Q2}. The picture of the molecular edges together with a piece of the bifurcation diagram is shown in Fig.~\ref{figtomolQ2xn}.
In this case the globally defined $\lambda$-cycles are $e_2\Gamma_2$. Since the connected components of the loop molecule differ exactly by the sign $e_2$, this cycle is fixed for the whole component and is not affected by any identifications in the pre-image of $L_0$.

\begin{figure}[!ht]
%figure17%
\centering
\includegraphics[width=0.55\textwidth, keepaspectratio = true]{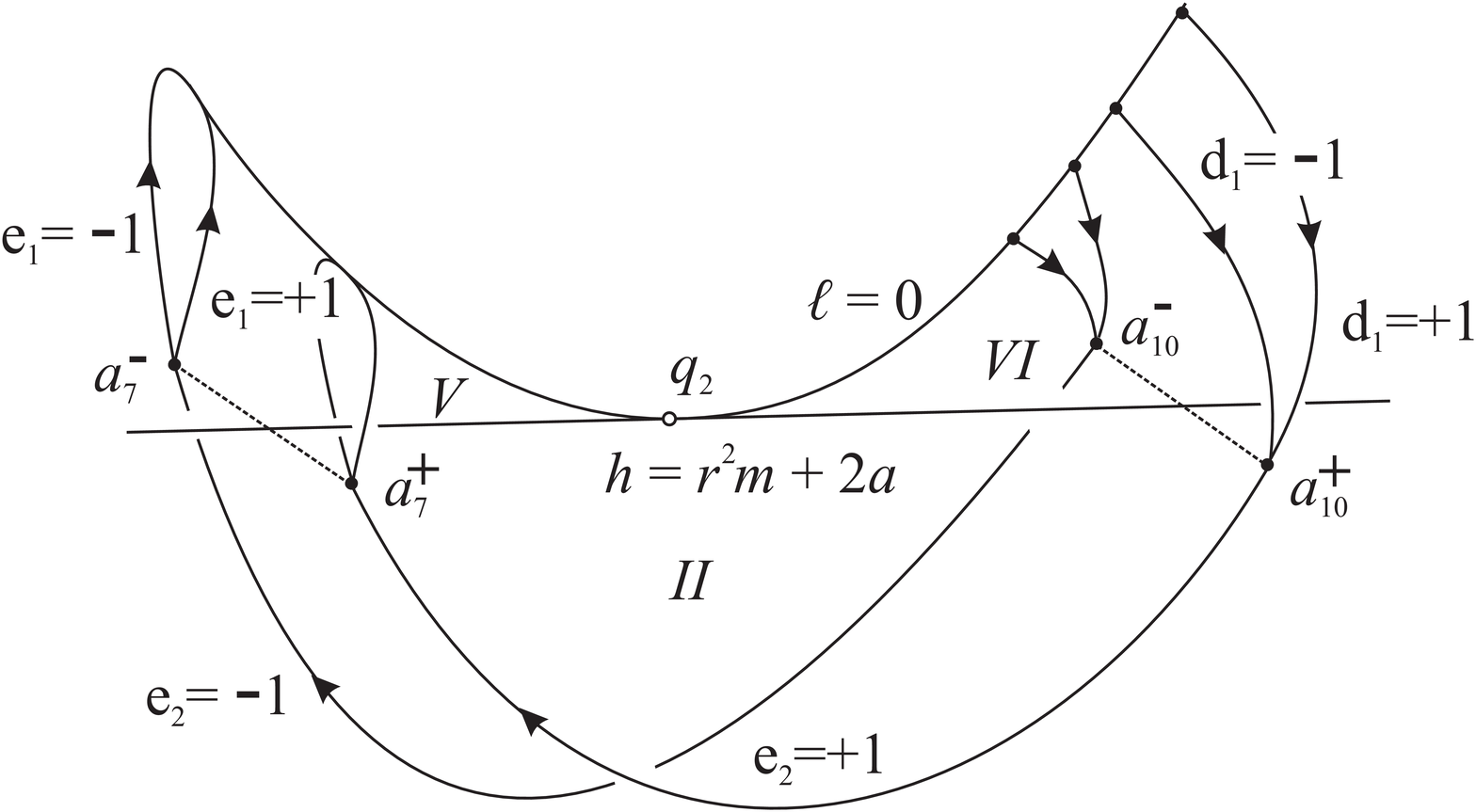}
\caption{Constructing the loop molecule of the point $q_2$.}\label{figtomolQ2xn}
\end{figure}

The edges in chamber $\ts{II}$ give $\mathcal{C}_{10}\out \to \mathcal{B}_7\inn$.
The gluing matrix is $E_+$. In chamber $\ts{V}$ let us fix $e_2$ and shift the bases $\mathcal{B}_{7}\out$ to the point on $L_0$. For $e_1=\pm 1$, we get two bases
$(e_2\Gamma_2,-\Gamma_1)$ and $(e_2\Gamma_2,\Gamma_1)$.
Previously we stated that, in chamber $\ts{V}$, $\chi_1(\Gamma_1)=-\Gamma_1$ and $\chi_2(\Gamma_2)=\Gamma_2$, therefore the obtained bases at the point of $L_0$ give the same basis on the image torus in $\mathcal{N}$. Thus, the gluing matrix on the loops is the identity matrix.

Finally, comparing the bases for segments $10$ and $13$ from Table~4 we see that the gluing matrix is $C_s$ where $s=\sgn ({-e_2 d_1})$. The theorem is proved.
\end{proof}

Note that the structure of the molecule of $q_1$ in Fig.~\ref{figlmolQ1Q2} allows us to pass the whole molecule in one direction. Instead of going globally right-to-left, let us change the direction in the upper half. Then we have to change the admissible bases on the boundary tori of the atoms $a_i^+$. The matrix $E_+$ remains unchanged but all the identity matrices turn into $E_-$. So, in this case we can avoid gluing matrices with the determinant equal to $+1$ and the loop manifold is orientable. Obviously, this cannot be done for the molecules of the point $q_2$ and for some similar iso-integral molecules having loops rather than arcs.

Let us denote by $P^2+ k\,g + s\,m$ the result of gluing $k$ handles and $s$ M\"{o}bius bands to a closed 2-surface $P^2$ (of course, it supposes cutting out small discs first). From Theorem~\ref{theo3} we find the topology of the loop manifolds of the degenerate points.

\begin{corol}\label{thecorol}
The loop manifold of the point $q_1$ is homeomorphic to the direct product $(\Cir^2+3g){\times}\Cir^1$, i.e., to the loop manifold of the non-degenerate \mbox{3-a}tom~$C_2$. The loop manifold of the point $q_2$ is homeomorphic to two copies of the direct product $(\Cir^2+4 m){\times}\Cir^1$.
\end{corol}

Indeed, all $r$-marks in the molecules of $q_1,q_2$ are equal to $\infty$ and the loop manifolds are easily restored from the topology of the graph. This fact corresponds to the mentioned above property of globally defined $\lambda$-cycles (at the point $q_1$ two $\lambda$-cycles $e_1\Gamma_1$ are identified by the $\mathbb{Z}_2$-symmetry $\chi$ in the pre-image of $L_0$).

\section{The collection of iso-integral molecules}\label{sec6}
In this section, we present the result of constructing all possible iso-$M$ and iso-energy marked molecules which occur in the system $\mathcal{N}$, except for the critical values of the integrals at four singular points ${h=\pm a \pm b}$, ${m=\pm 1/(a\pm b)}$ and, for the restriction of $H$ to $\mathcal{N}$, the minimal value ${h_*=\sqrt{2(a^2+b^2)}}$ of the $h$-coordinate on the curve $L_0$.

\def\figs{0.24}
{
\renewcommand{\rul}{\rule[-5pt]{0pt}{14pt}}
\footnotesize
\tabcolsep=0mm
%\begin{longtable}{|m{6mm}|m{29mm}||m{6mm}|m{29mm}||m{6mm}|m{29mm}|}%
\begin{longtable}{|m{8mm}|m{38mm}||m{8mm}|m{38mm}||m{8mm}|m{38mm}|}
\multicolumn{6}{l}{\fts{Table 5. Molecules and their codes}}\\
\hline
\rul \hbox to 7mm {\hfil\fns{Code}\hfil}
&
\rul \hbox to 36mm {\hfil\fns{Molecule}\hfil}
&
\rul \hbox to 7mm {\hfil\fns{Code}\hfil}
&
\rul \hbox to 36mm {\hfil\fns{Molecule}\hfil}
&
\rul \hbox to 7mm {\hfil\fns{Code}\hfil}
&
\rul \hbox to 36mm {\hfil\fns{Molecule}\hfil} \\
\hline\endfirsthead%
\multicolumn{6}{r}{\fts{Table 5 (contuned)}}\\
\hline
\rul \hbox to 7mm {\hfil\fns{Code}\hfil}
&
\rul \hbox to 36mm {\hfil\fns{Molecule}\hfil}
&
\rul \hbox to 7mm {\hfil\fns{Code}\hfil}
&
\rul \hbox to 36mm {\hfil\fns{Molecule}\hfil}
&
\rul \hbox to 7mm {\hfil\fns{Code}\hfil}
&
\rul \hbox to 36mm {\hfil\fns{Molecule}\hfil} \\
\hline\endhead

$\ts{OM}_1$ & \includegraphics[width=\figs\textwidth, keepaspectratio = true]{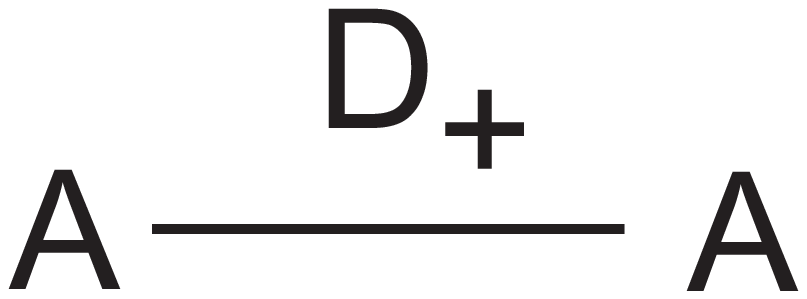}
&
$\ts{OM}_2$ & \includegraphics[width=\figs\textwidth, keepaspectratio = true]{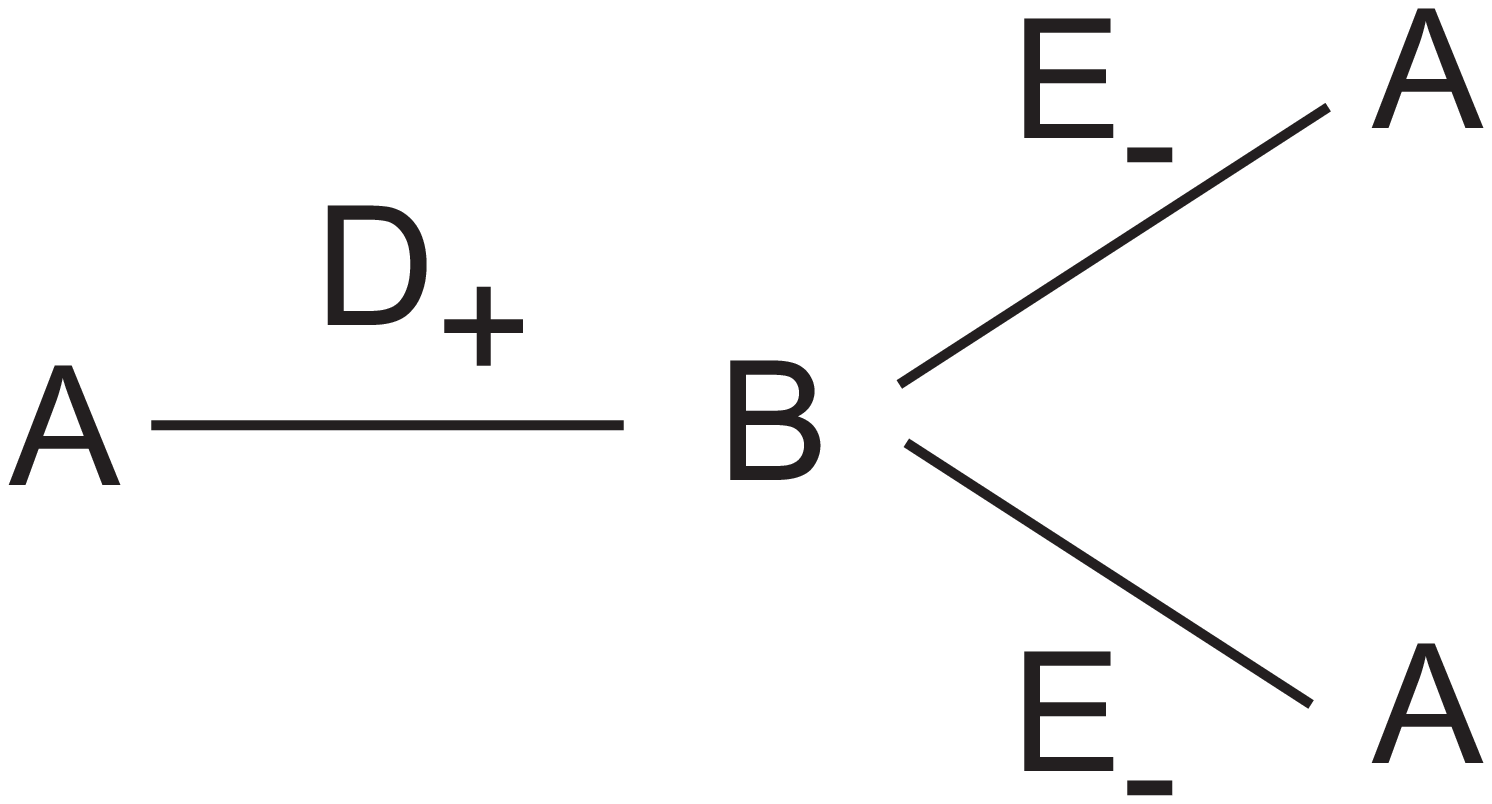}
&
$\ts{OM}_3$& \includegraphics[width=\figs\textwidth, keepaspectratio = true]{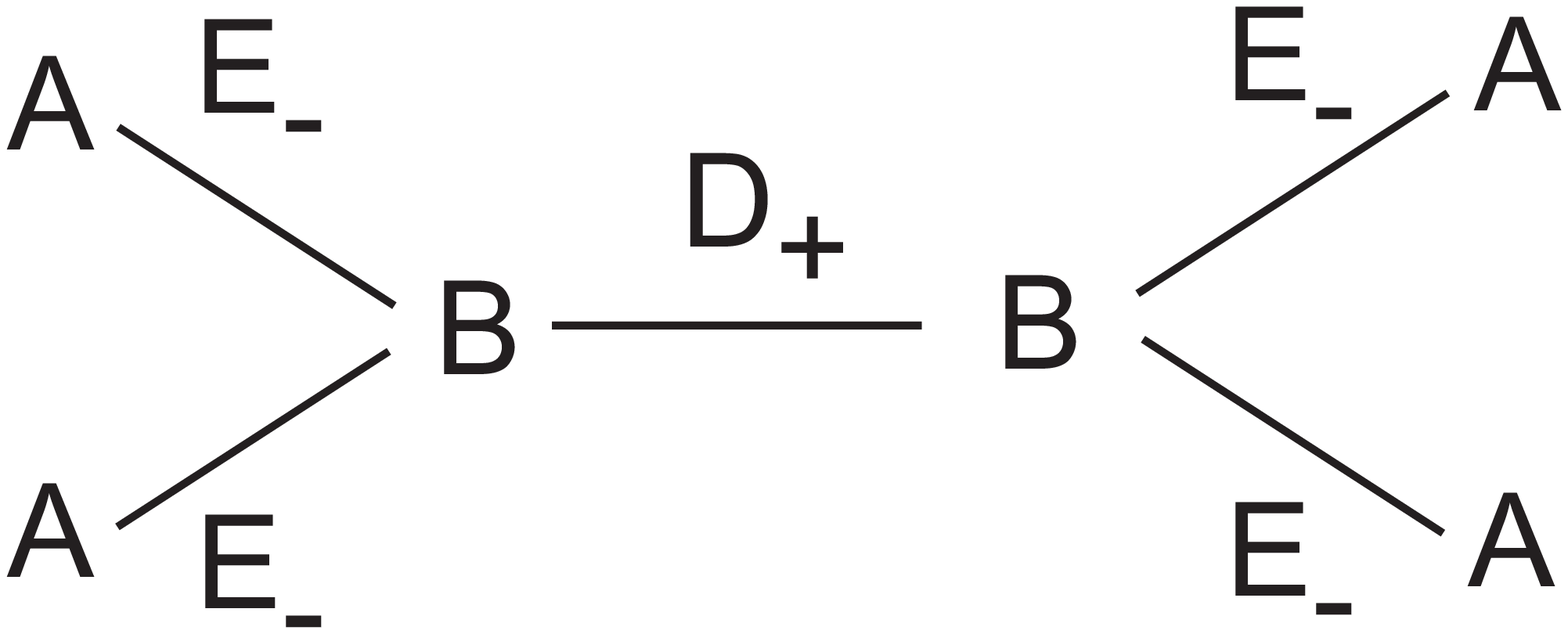}\\
\hline

$\ts{OM}_4$ & \includegraphics[width=\figs\textwidth, keepaspectratio = true]{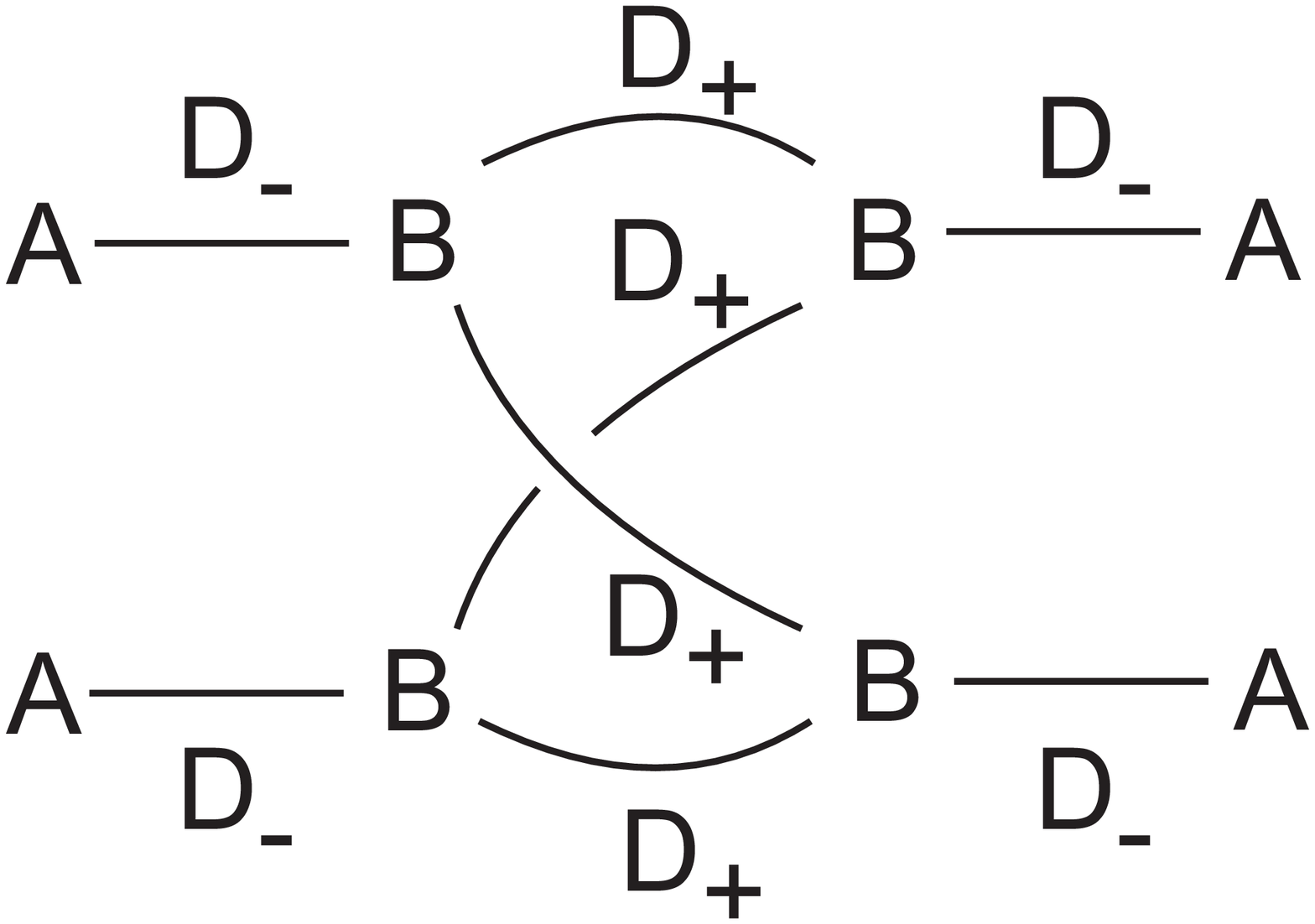}
&
$\ts{OM}_5$ & \includegraphics[width=\figs\textwidth, keepaspectratio = true]{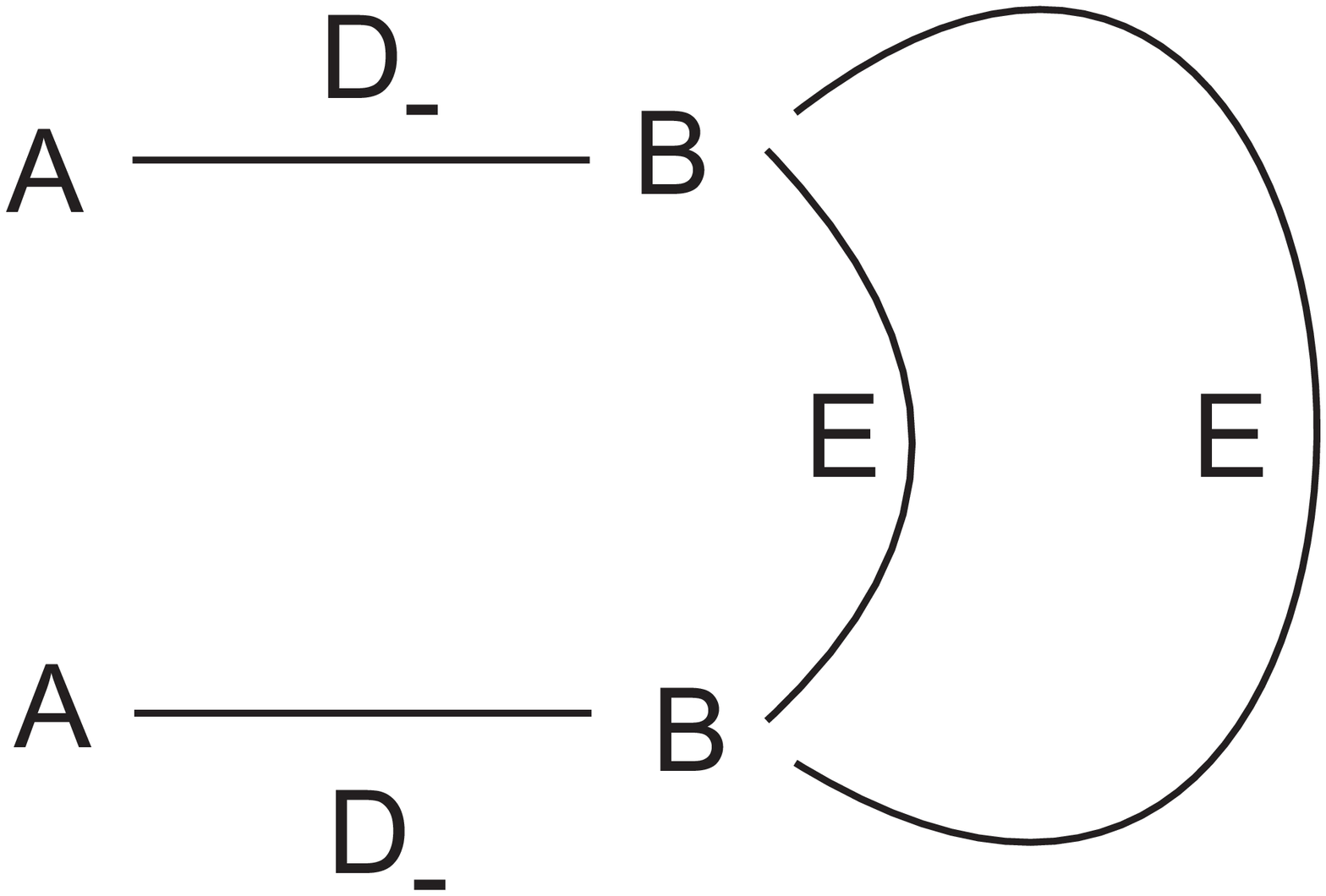}
&
$\ts{NM}_1$& \includegraphics[width=\figs\textwidth, keepaspectratio = true]{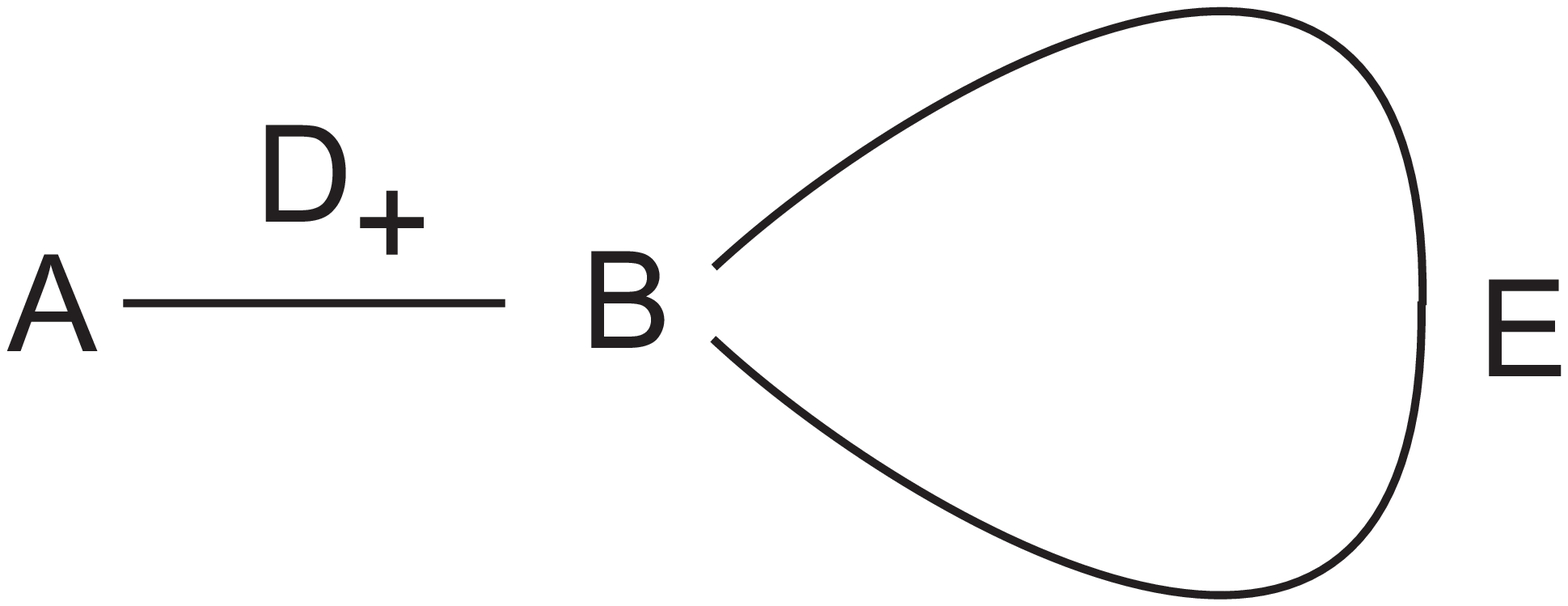}\\
\hline

$\ts{NM}_2$ & \includegraphics[width=\figs\textwidth, keepaspectratio = true]{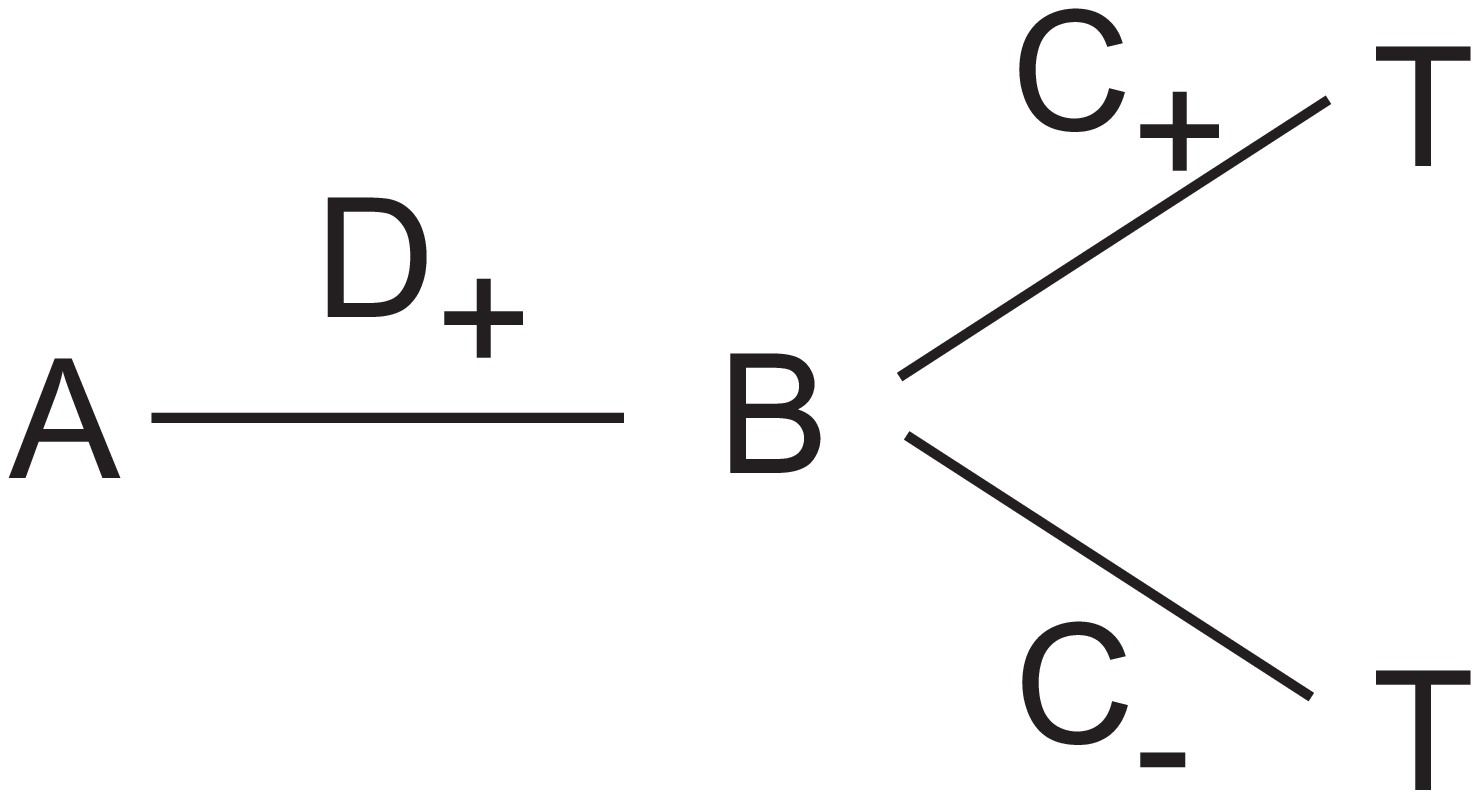}
&
$\ts{NM}_3$ & \includegraphics[width=\figs\textwidth, keepaspectratio = true]{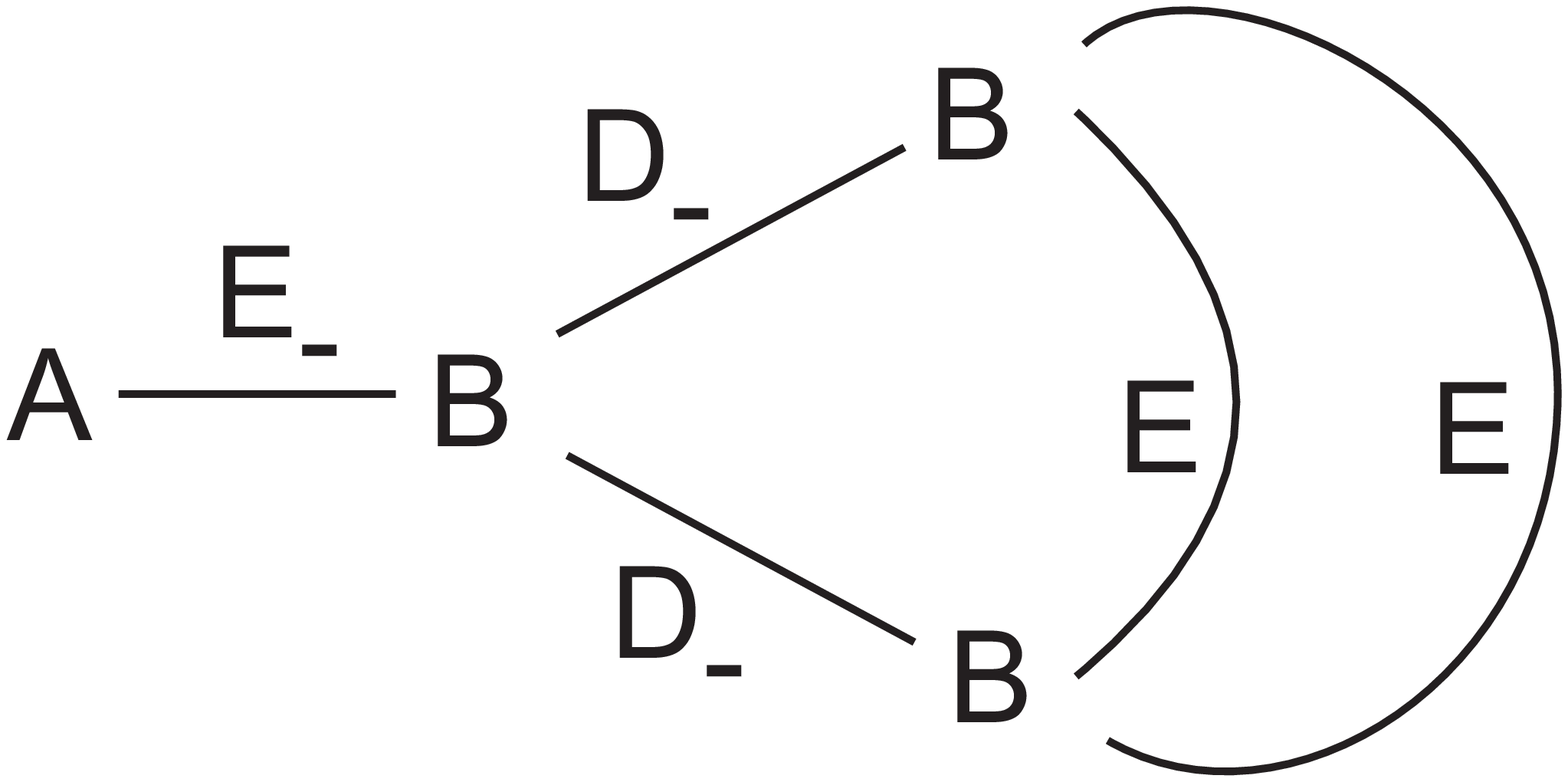}
&
$\ts{NM}_4$& \includegraphics[width=\figs\textwidth, keepaspectratio = true]{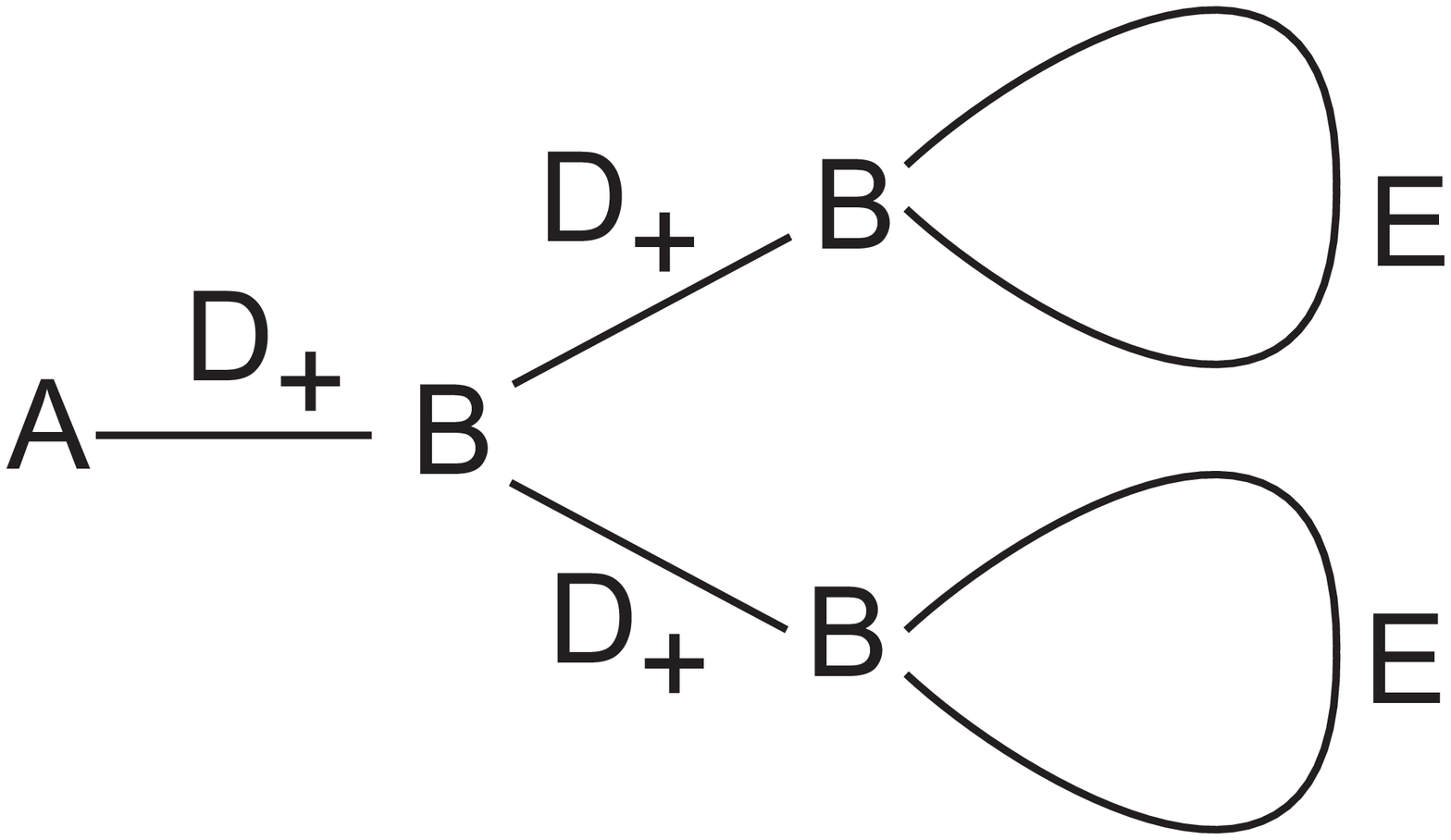}\\
\hline

$\ts{NM}_5$ & \includegraphics[width=\figs\textwidth, keepaspectratio = true]{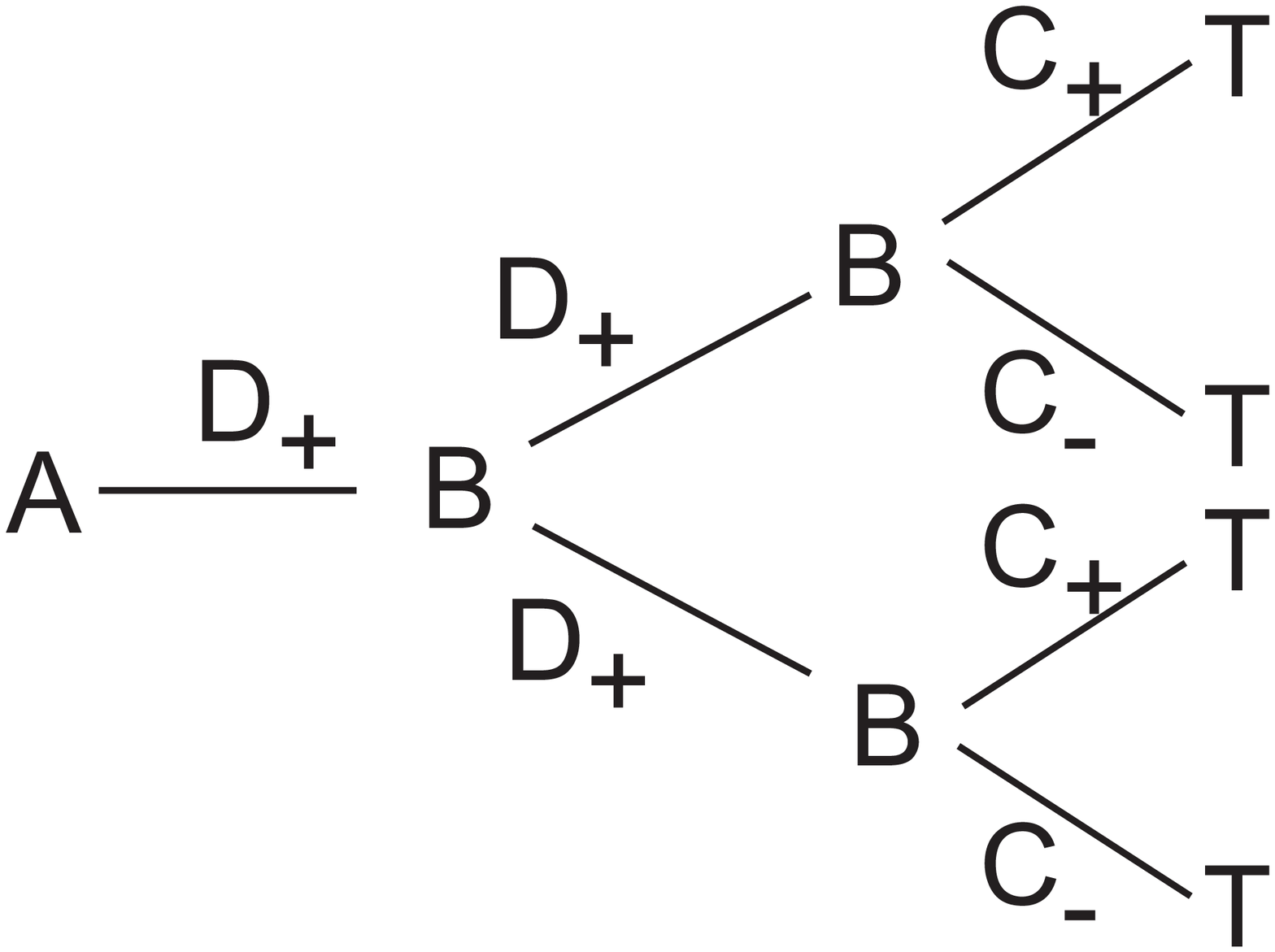}
&
$\ts{UM}_1$ & \includegraphics[width=\figs\textwidth, keepaspectratio = true]{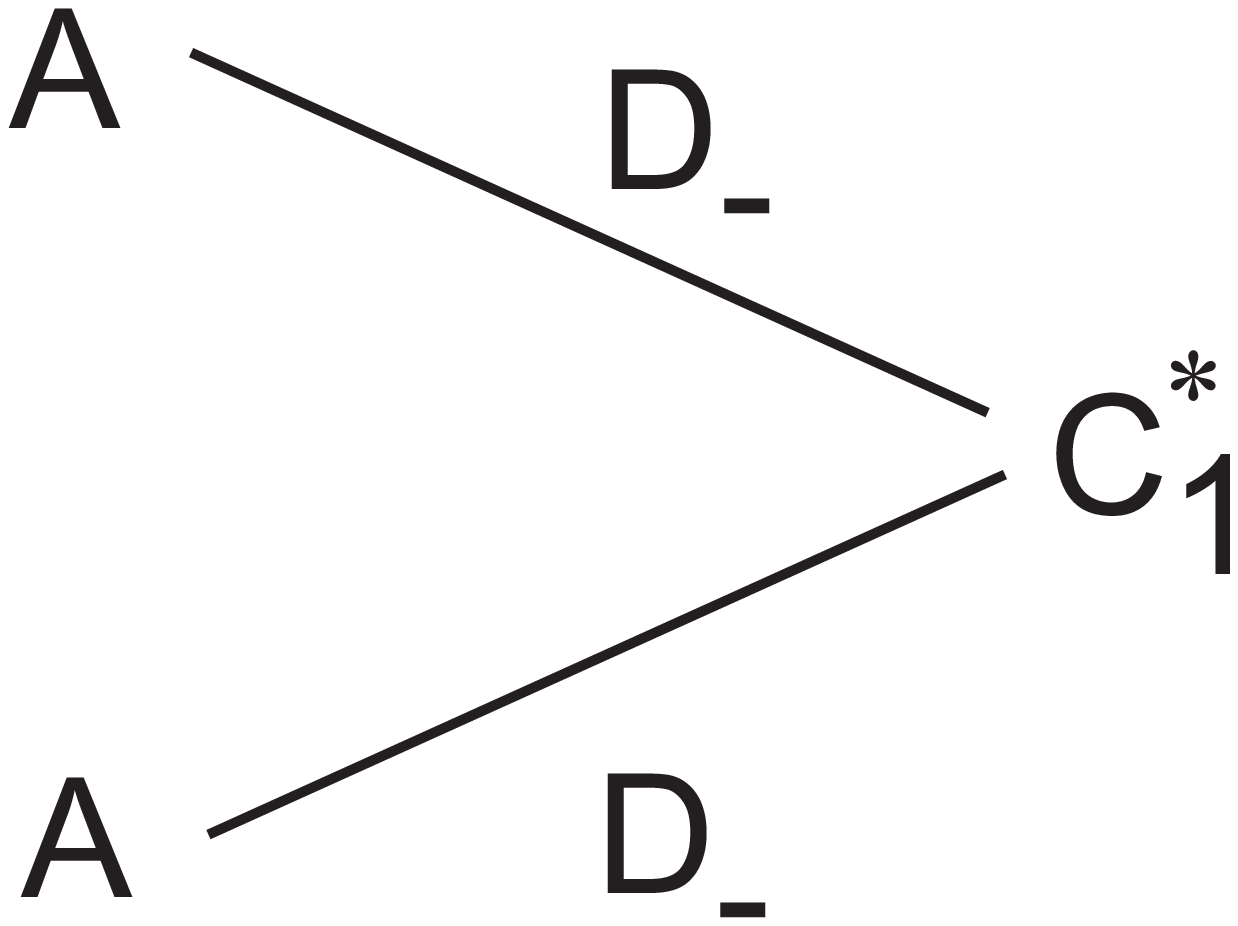}
&
$\ts{UM}_2$& \includegraphics[width=\figs\textwidth, keepaspectratio = true]{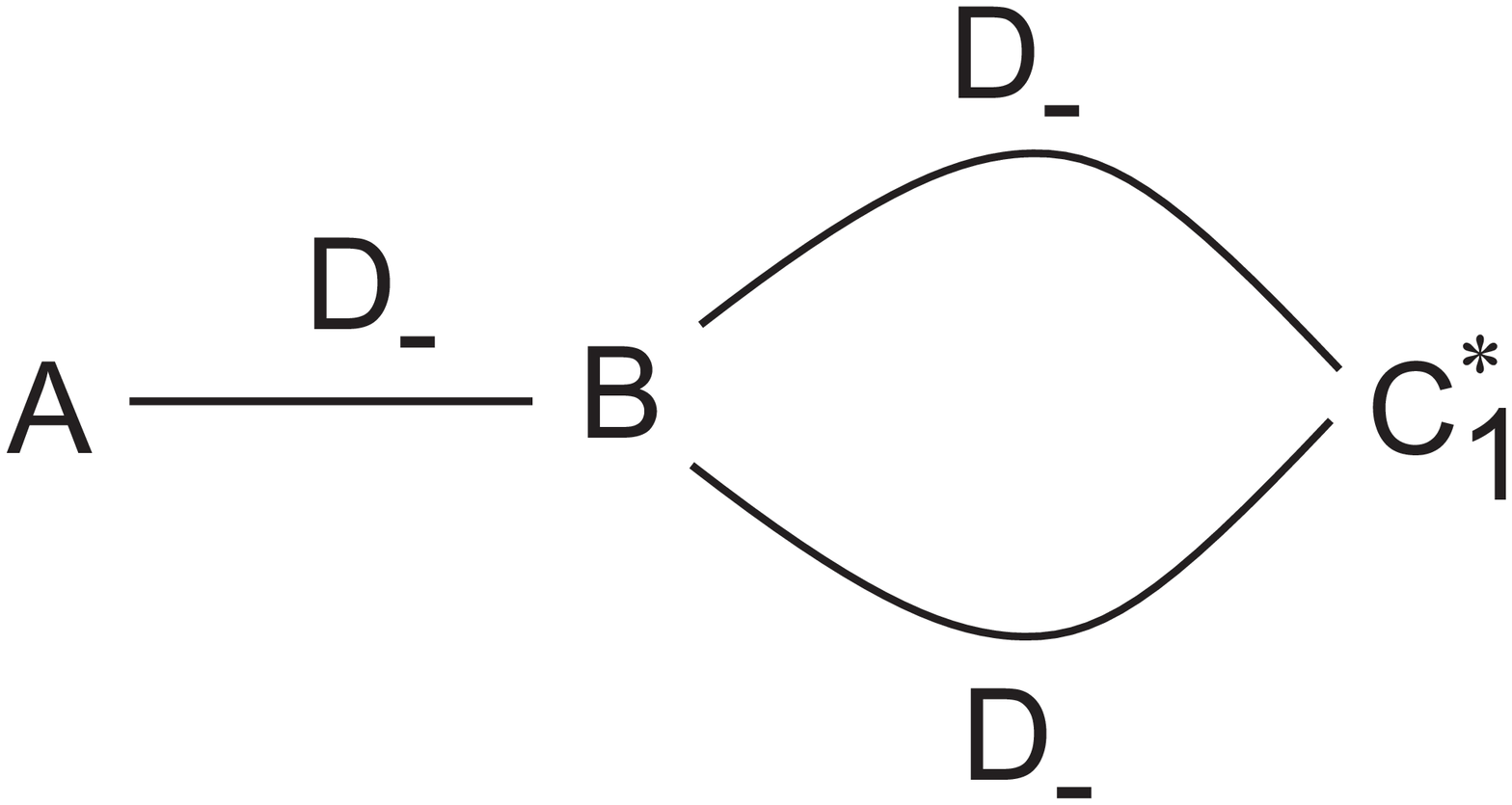}\\
\hline

$\ts{UM}_3$ & \includegraphics[width=\figs\textwidth, keepaspectratio = true]{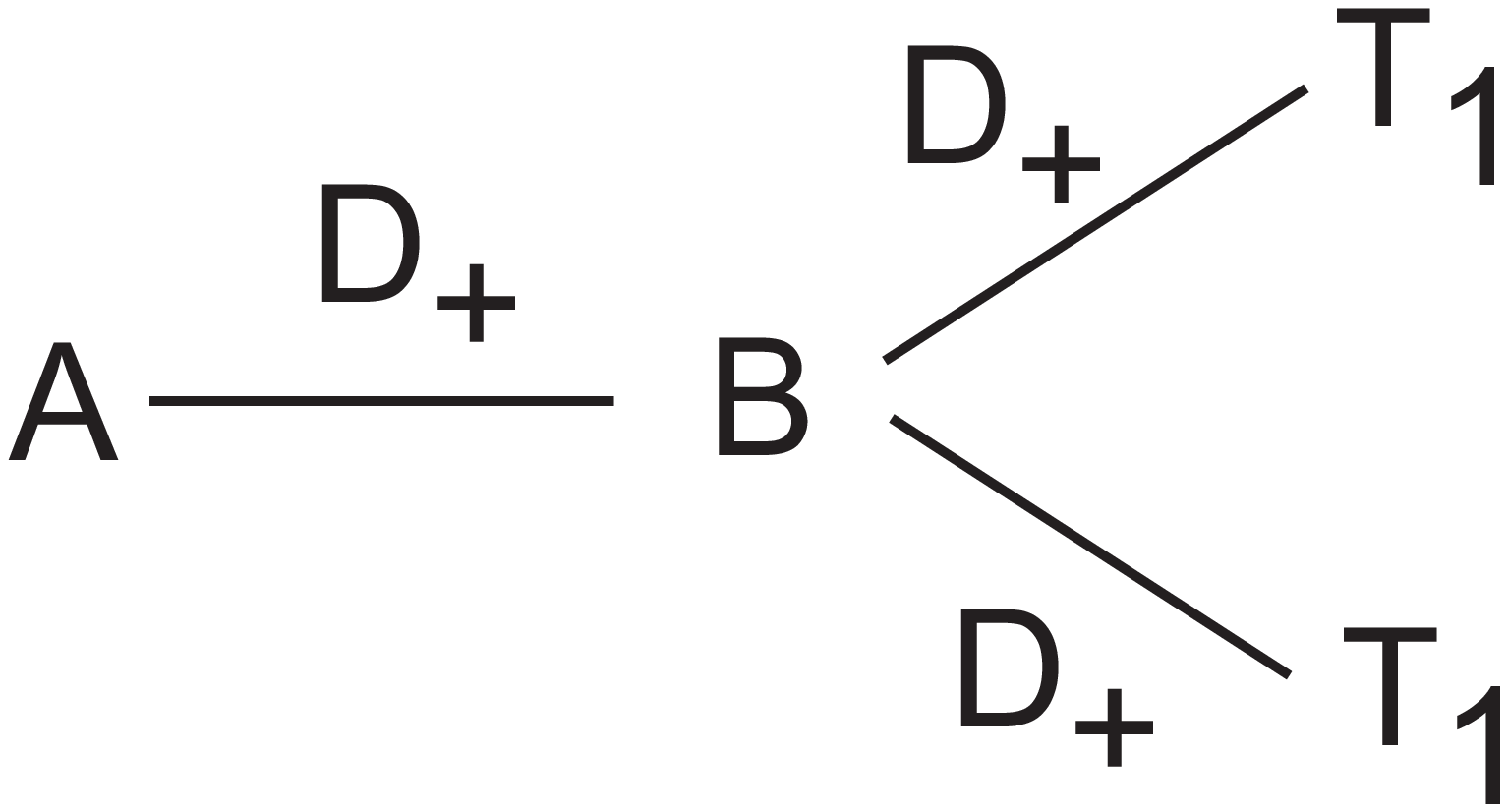}
&
$\ts{UM}_4$ & \includegraphics[width=\figs\textwidth, keepaspectratio = true]{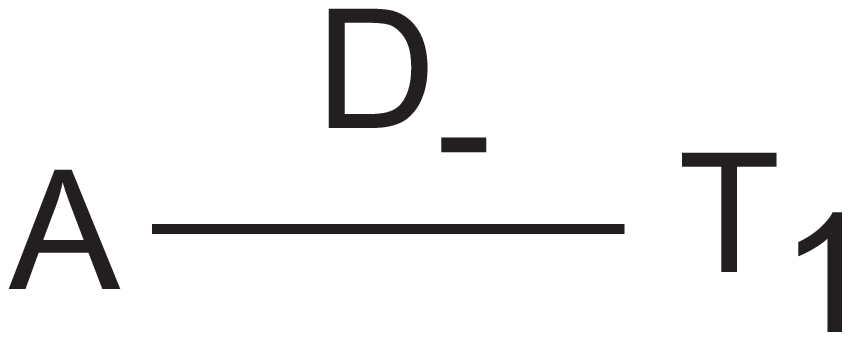}
&
$\ts{LM}_0$& \includegraphics[width=\figs\textwidth, keepaspectratio = true]{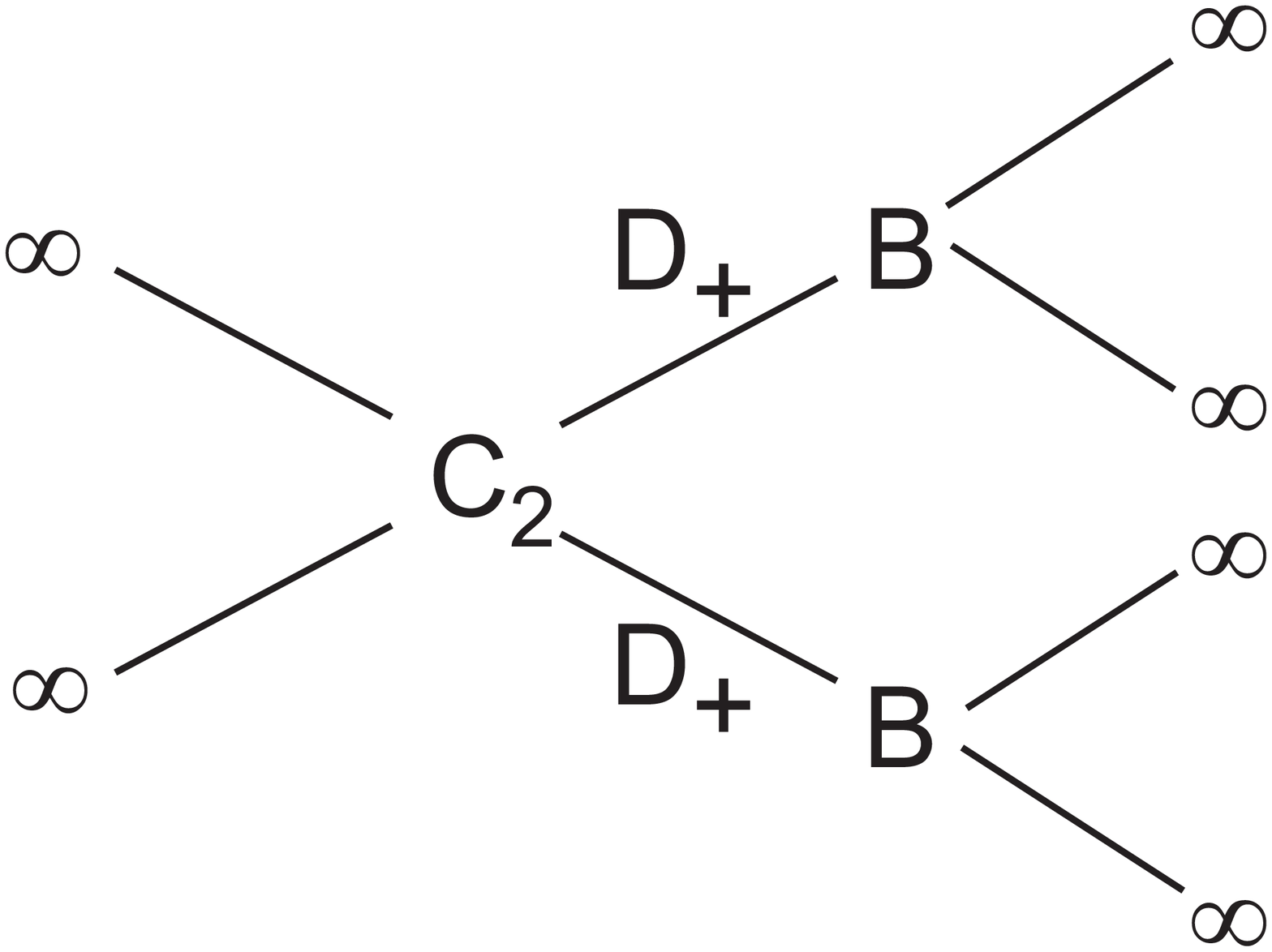}\\
\hline
\hspace*{1mm}$\ts{IM}_1$ & \includegraphics[width=\figs\textwidth, keepaspectratio = true]{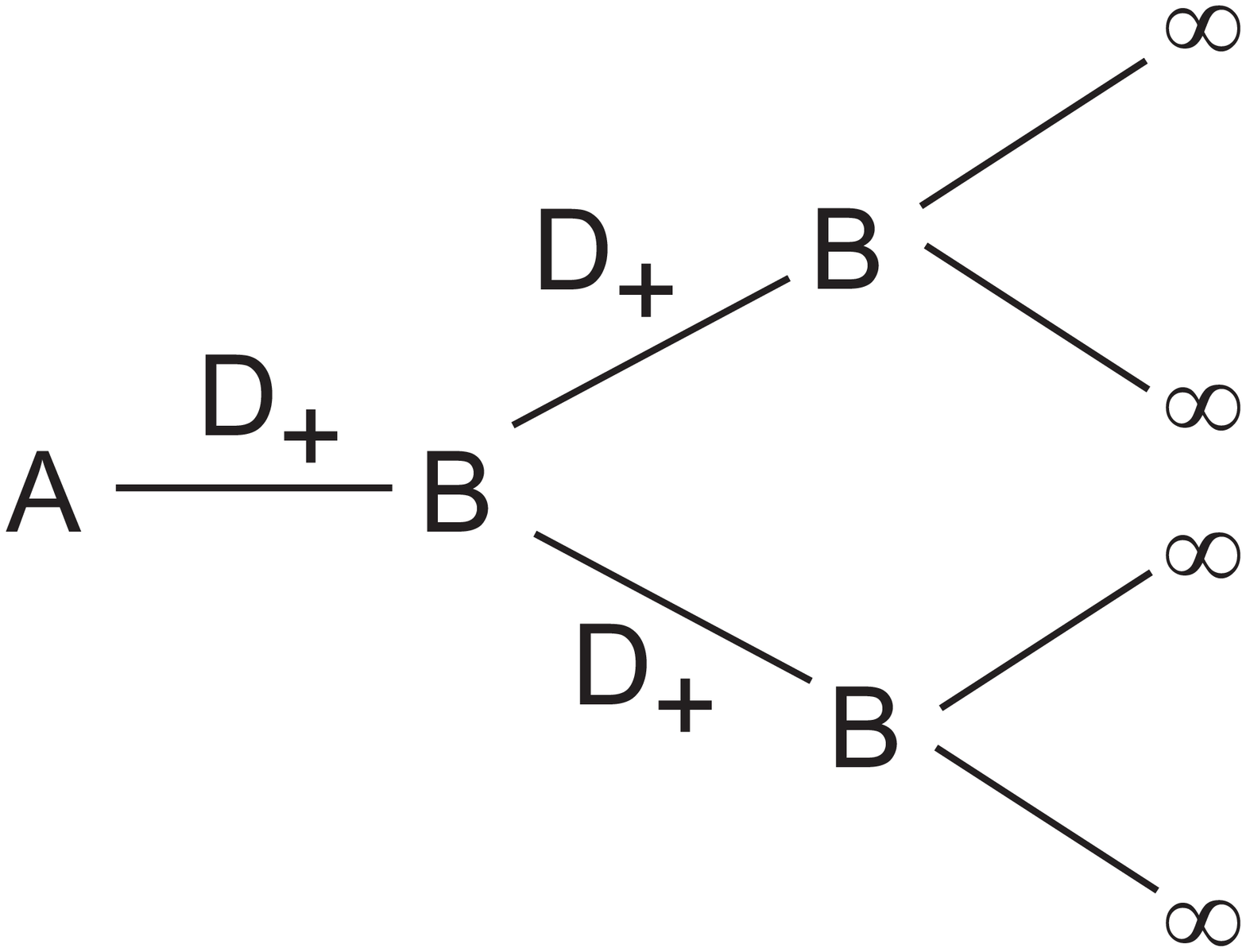}
&
\hspace*{1mm}$\ts{IM}_2$ & \includegraphics[width=\figs\textwidth, keepaspectratio = true]{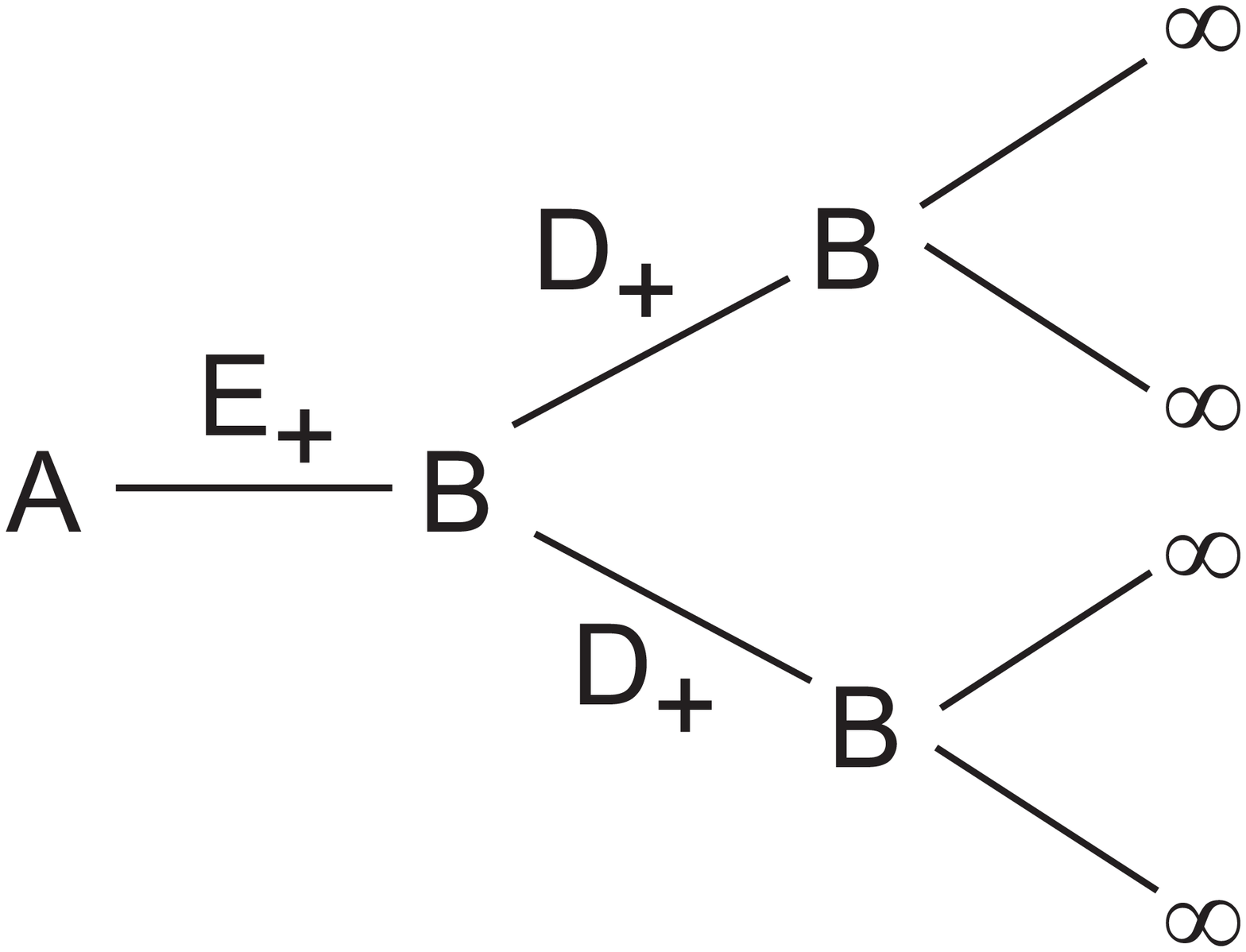}
&
\hspace*{1mm}$\ts{IM}_3$& \includegraphics[width=\figs\textwidth, keepaspectratio = true]{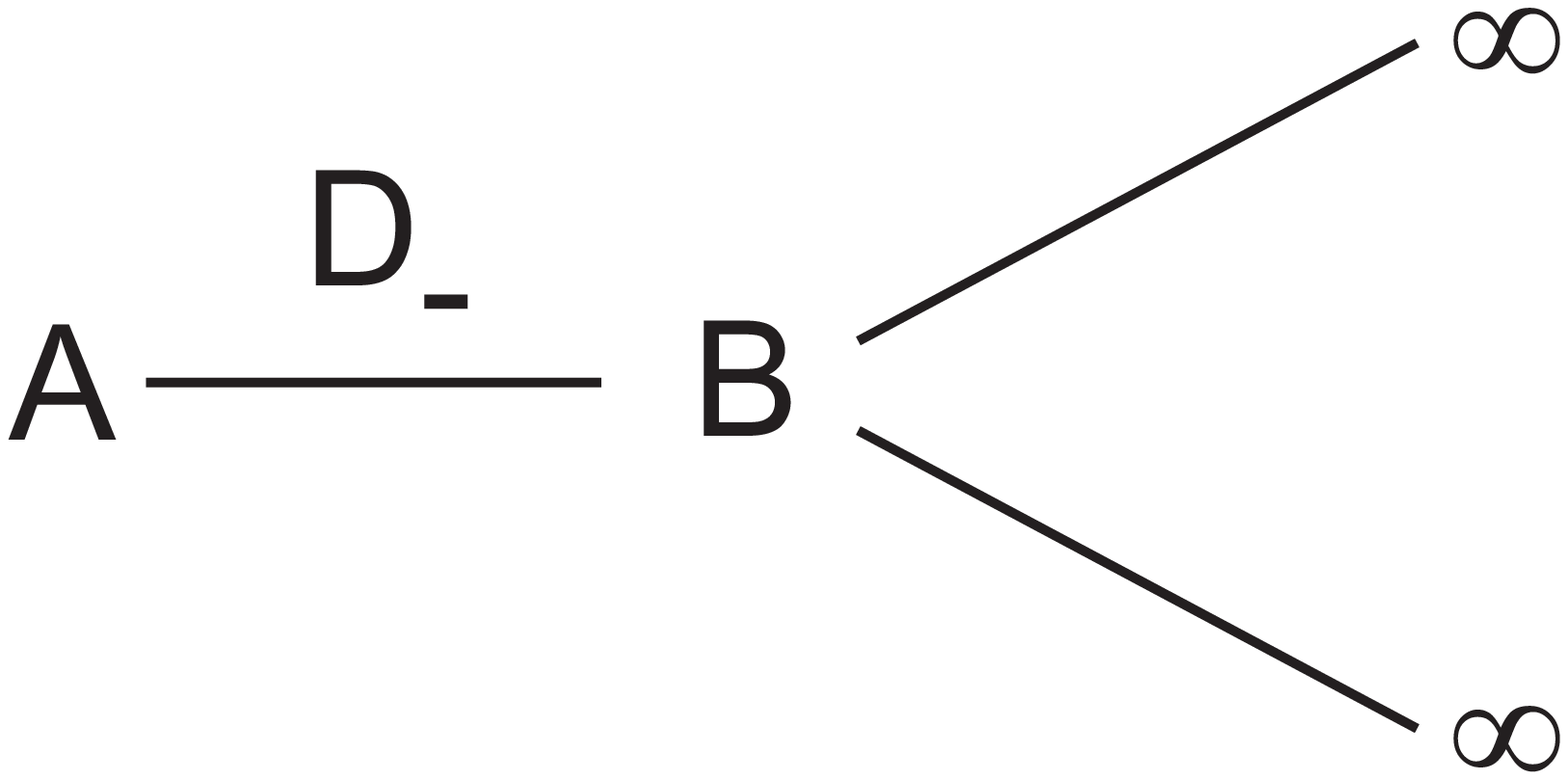}\\
\hline
\end{longtable}
}

In Table~5, the notation $\ts{OM}_i$ stands for orientable molecules, $\ts{NM}_i$ for non-orientable ones, $\ts{UM}_i$ for splitting ones. We also include the non-compact iso-$L$ molecule $\ts{LM}_0$ for the level $\{L=0\}$ where the symplectic structure degenerates. For non-negative values of $m$ the iso-$M$ manifolds are non-compact. Thus, crossing the zero value, these manifolds bifurcate without any critical points of $M$ appearing on the zero level. Three corresponding non-compact molecules are denoted by $\ts{IM}_i$. Note that the molecule $\ts{IM}_1$ was found in \cite{Zo2006} as describing the zero level of the Bogoyavlensky integral. This level, due to \eqref{eq2_1}, coincides with the manifold $\{M=0\}\subset \mathcal{N}$.

The molecules with the gluing matrices are representatives of their Liouville equivalence classes obtained by our way of choosing the coordinate systems and orientations on the families of tori. In \cite{igs}, one can find the description of possible changes that occur due to the changes of directions and orientations. Finally, in Table~6 we show the correspondence between the values of $m$ and $h$ and the molecules found.

{

\centering
\small
\renewcommand{\arraystretch}{1.3}%
\begin{longtable}{|c|c||c|c|}
\multicolumn{4}{l}{\fts{Table 6. Iso-integral molecules}}\\
\hline
\fns{$m$-value} & \fns{Code} & \fns{$h$-value} & \fns{Code} \\
\hline\endfirsthead%
\multicolumn{4}{r}{\fts{Table 6 (continued)}}\\
\hline
\fns{$m$-value} & \fns{Code} & \fns{$h$-value} & \fns{Code} \\
\hline\endhead
$m<\min \{m(q_1),m(p_3)\}$ & $\ts{OM}_5$ & {$h(p_1)<h<h(p_2)$} &{$\ts{OM}_1$}\\
\hline
$m(q_1)<m <m(p_3)\; (a>3b)$ & $\ts{OM}_5$ & {$h(p_2)<h<h(p_3)$} &{$\ts{OM}_2$}\\
\hline
$m(p_3)<m <m(q_1)\; (a<3b)$ & $\ts{NM}_3$ & {$h(p_3)<h<h(p_4)$} &{$\ts{OM}_3$}\\
\hline
$m=m(q_1) <m(p_3)\; (a>3b)$ & $\ts{UM}_1$ & {$h(p_4)<h<h_*$} &{$\ts{OM}_4$}\\
\hline
$m=m(q_1)> m(p_3)\; (a<3b)$ & $\ts{UM}_2$ & {$h_*<h<h(q_2)$} &{$\ts{OM}_5+2\ts{NM}_1$}\\
\hline
$\max\{m(p_3),m(q_1)\}< m < m(p_4)$ & $\ts{NM}_3$ & {$h=h(q_2)$} &{$\ts{OM}_5+2\ts{UM}_4$}\\
\hline
$m(p_4)<m <m(q_2)$ & $\ts{NM}_4$ & {$h(q_2)<h<h(q_1)$} &{$\ts{OM}_5+2\ts{NM}_2$}\\
\hline
$m=m(q_2)$ & $\ts{UM}_3$ & {$h=h(q_1)$} &{$\ts{UM}_1+2\ts{NM}_2$}\\
\hline
$m(q_2)<m <0 $ & $\ts{NM}_5$ & {$h>h(q_1)$} &{$\ts{OM}_5+2\ts{NM}_2$}\\
\hline
$0 \leqslant m < m(p_1)$ & $\ts{IM}_1$ & {} &{}\\
\hline
$m(p_1)< m <m(p_2)$ & $\ts{IM}_2$  & {} &{}\\
\hline
$ m > m(p_2)$ & $2\ts{IM}_3$ & {} &{}\\
\hline
\end{longtable}

}

This completes the exact topological analysis of the system.

\section*{Acknowledgements}\addcontentsline{toc}{section}{Acknowledgements}
The author is grateful to A.V.\,Bolsinov for extremely valuable discussions and advices, to V.N.\,Roubtsov, the University of Angers and the Organizers of FDIS-2013 for hospitality and support.

%\clearpage

\end{document}